\documentclass{llncs}
\usepackage{amsmath}
\usepackage{amssymb}
\usepackage{bm}  

\usepackage[a4paper, total={7in, 9in}]{geometry}
\usepackage{graphicx}
\usepackage{float}
\usepackage{array}
\usepackage{booktabs}
\usepackage{longtable}
\usepackage{tabularx}

\usepackage{xcolor}
\usepackage{soul}

\usepackage{hyperref}

\usepackage[font=small,skip=12pt]{caption}
\usepackage{subcaption}

\usepackage{algorithm}
\usepackage{algpseudocodex}  

\usepackage{listings}

\usepackage{tikz}
\usetikzlibrary{shapes.geometric, arrows}
\usepackage{academicons}
\usepackage{enumerate}
\usepackage{lineno}

\makeatletter
\def\algbackskip{\hskip-\ALG@thistlm}
\makeatother

\AtBeginDocument{%
 \providecommand\BibTeX{{%
   \normalfont B\kern-0.5em{\scshape i\kern-0.25em b}\kern-0.8em\TeX}}}

\DeclareUnicodeCharacter{2212}{\ensuremath{-}}

\begin{document}
\title{Area-universality in Outerplanar Graphs}
\author{Ravi Suthar, Raveena, Krishnendra Shekhawat}
\institute{Department of Mathematics, Birla Institute of Technology and Science, Pilani, Pilani Campus, Vidya Vihar, Pilani, Rajasthan 333031, India}

\maketitle
 \begin{abstract}
A rectangular floorplan is a partition of a rectangle into smaller rectangles such that no four rectangles meet at a single point. Rectangular floorplans arise naturally in a variety of applications, including VLSI design, architectural layout, and cartography, where efficient and flexible spatial subdivisions are required. A central concept in this domain is that of area-universality: a floorplan (or more generally, a rectangular layout) is area-universal if, for any assignment of target areas to its constituent rectangles, there exists a combinatorially equivalent layout that realizes these areas.

In this paper, we investigate the structural conditions under which an outerplanar graph admits an area-universal rectangular layout. We establish a necessary and sufficient condition for area-universality in this setting, thereby providing a complete characterization of admissible outerplanar graphs. Furthermore, we present an algorithmic construction that guarantees that the resulting layout is always area-universal.
\end{abstract}




\section{Introduction}
The rectangular layout is a fundamental concept that involves strategic arrangements of non-overlapping rectangles (modules) within a rectangular space. Additionally, it ensures that there are no intersecting joints formed by any four non-overlapping modules within the layout. The rectangular layout has broad application in cartography \cite{raisz1934rectangular}, architecture \cite{earl1979architectural}, VLSI circuit design \cite{yeap1995sliceable}, and graph drawings. This paper explores the idea of area-universal rectangular layouts using a graph-theoretic approach. The notion of area-universal layouts \cite{eppstein2009area} emerges as an essential instrument for generating diverse rectangular arrangements that accommodate specific area assignments while preserving the consistent adjacency configurations defined by the adjacency graphs. From a mathematical perspective, the area-universal layouts can be represented through graphs, wherein each room or module within the layout is represented as a vertex, and the adjacency (i.e., the physical connection or proximity between rooms) is denoted by an edge connecting the corresponding vertices. The central challenge lies in determining whether the graph that captures these adjacency constraints can accommodate every possible assignment of areas to its rectangles while preserving the adjacencies and connectivity among the modules. Area-universal rectangular layouts of graphs serve a vital role in various fields \cite{mumford2008drawing,bruls2000squarified}, particularly in VLSI design, where the rectangles represent circuit components and their common boundaries model adjacency requirements. During the early stages of VLSI design, when chip component areas are not yet specified, only the relative positions of the components matter. However, later design stages require specific component areas, which an area-universal layout accommodates, simplifying the design process. Past works \cite{van2007rectangular,wimer2002floorplans} have presented heuristic algorithms for computing rectangular layouts based on given area assignments.


\subsection{Preliminary}
A graph $\mathcal{G} = (V, E)$ is a fundamental combinatorial structure consisting of a finite set of vertices $V$ and a set of edges $E$, where each edge connects a pair of distinct vertices in $V$.
A graph $\mathcal{G}$ is \emph{biconnected} if, for every pair of distinct vertices $u, v \in V(\mathcal{G})$, there exist at least two vertex-disjoint paths between $u$ and $v$. A graph is called planar if it can be drawn on the plane in such a way that its edges intersect only at their shared endpoints, thereby allowing a crossing-free representation. Such a drawing of a planar graph in the plane without edge crossings is referred to as a plane graph or a plane embedding. In this paper, the \emph{order} of a graph $\mathcal{G}$ denotes $|V|$, the number of vertices of $\mathcal{G}$.

\begin{definition}
In a plane graph, a \emph{separating triangle}~\cite{sun1993edge} is a cycle of 
three edges $(u,v,w)$ such that at least one vertex of the graph lies in the 
interior region bounded by this cycle (see Figure \ref{ST}). 
\end{definition}

\begin{definition}
        A \textit{properly triangulated plane graph} (PTPG)  \cite{kozminski1985rectangular} is a connected plane graph which satisfies the following properties: (see Figure \ref{ptpg})
        \begin{enumerate}
             \item Exterior face has length $ \geq 4$,
             \item Every face (except the exterior) is a triangle (bounded by three edges),
             \item It does not contain a separating triangle (ST).
        \end{enumerate}
        \begin{figure}
         \centering
         \begin{subfigure}[b]{0.2\textwidth}
             \centering
             \includegraphics[width=\textwidth]{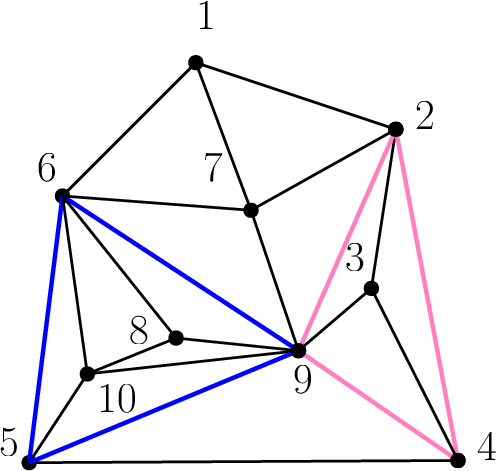}
             \caption{ }
             \label{ST}
         \end{subfigure}
         \hspace{1cm}
         \begin{subfigure}[b]{0.23\textwidth}
             \centering
             \includegraphics[width=\textwidth]{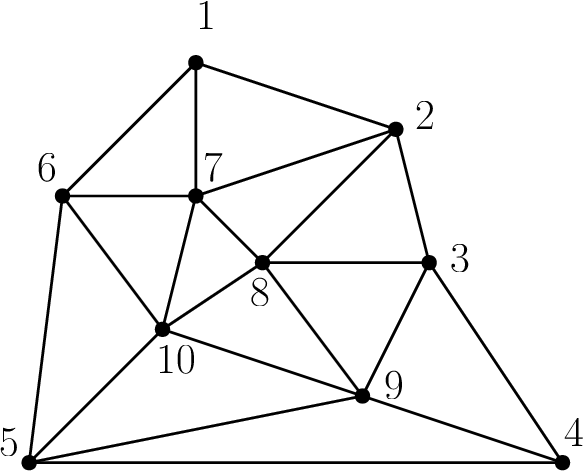}
             \caption{ }
             \label{ptpg}
         \end{subfigure}
         \hspace{1cm}
         \begin{subfigure}[b]{0.23\textwidth}
             \centering
             \includegraphics[width=\textwidth]{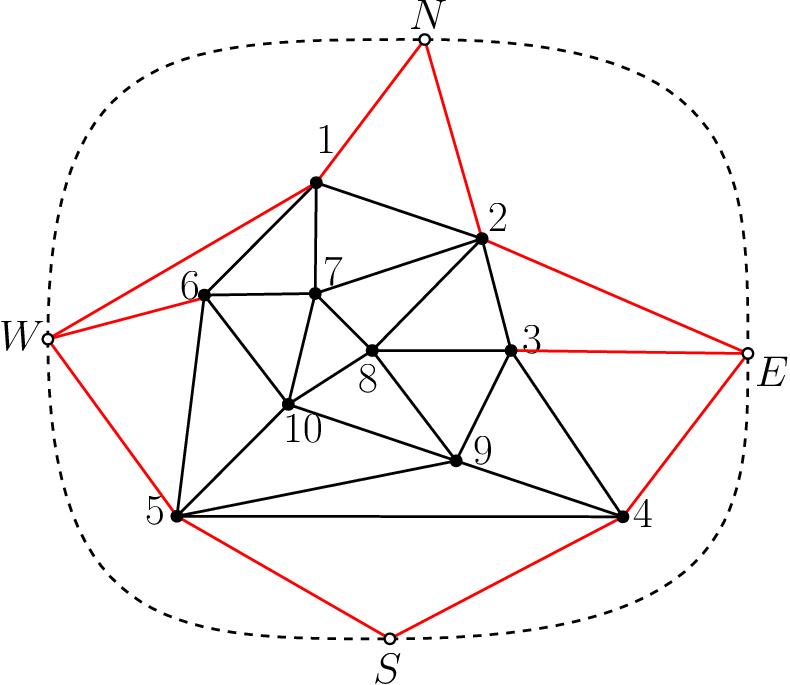}
             \caption{ }
             \label{4ptpg}
         \end{subfigure}
        
            \caption{(a) Presence of separating triangles $(2,4,9)$, $(6,9,10)$, and $(6,9,5)$. 
(b) A properly triangulated plane graph (PTPG). 
(c) An \emph{extended graph} of a PTPG.}
            \label{STPTPG}
\end{figure}
\end{definition}

\begin{definition}\label{def_outerplanar}
    A graph $\mathcal{G}$ is said to be \emph{outerplanar} if it admits a plane embedding in which 
    all vertices lie on the boundary of the exterior face.
\end{definition}
Let $\mathcal{G}$ be a plane triangulated graph. If $\mathcal{G}$ can be augmented by adding four external vertices, called the \emph{cardinal vertices} and denoted by $N, W, S,$ and $E$, placed on the outer face in clockwise order (see Figure \ref{4ptpg}), such that the resulting graph $E(\mathcal{G})$ satisfies the following properties:
(i) every interior face of $E(\mathcal{G})$ is a triangle and the exterior face is a quadrilateral; and
(ii) $E(\mathcal{G})$ contains no separating triangles, then we call $E(\mathcal{G})$ an \emph{extended graph} of $\mathcal{G}$. A plane triangulated graph $\mathcal{G}$ that admits such an augmentation is called a \emph{proper graph}. If, in addition, $\mathcal{G}$ is outerplanar, then the corresponding extended graph $E(\mathcal{G})$ is called an \emph{extended outerplanar graph} (see Figure~\ref{extended}).

\begin{definition}
    Regular edge labeling [REL] \cite{kant1997regular}: A regular edge labeling for a biconnected PTPG $\mathcal{G}$ with four outer vertices $N$, $W$, $S$, $E$ (ordered counterclockwise), constitutes a partition and orientation of its interior edges into two disjoint subsets, $T_1$ (represented using blue color) and $T_2$ (represented using red color), satisfying the following conditions:
    \begin{itemize}
        \item[(I)] For every interior vertex $v$, the incident edges are arranged counterclockwise around $v$ in the sequence:
        \begin{itemize}
            \item Directed toward $v$: incoming $T_1$ edges.
            \item Directed away from $v$: outgoing $T_2$ edges
            \item Directed away from $v$: outgoing $T_1$ edges
            \item Directed toward $v$: incoming $T_2$ edges
        \end{itemize}

        \item[(ii)] Edges incident to vertex N are included in $T_1$ and are directed toward N. Conversely, edges incident to W are part of $T_2$ and are directed away from W. Edges incident to S are contained in $T_1$ and are directed away from S, while edges incident to E lie in the set $T_2$ and are directed toward E.
    \end{itemize}
\end{definition}

\begin{definition}
     A flippable edge~\cite{eppstein2012area} in a REL is an edge $e$ that is not incident to any degree-four vertex, and must be a diagonal of a four-cycle whose edges are alternately labeled in the REL. Furthermore, if an edge is flippable, its label can be changed without affecting the rest of the labeling. In Figure \ref{flip}, edges $(5, 7)$ and $(4, 8)$ are flippable, and the REL obtained after flipping these edges is shown in Figure \ref{flip2}.
    \begin{figure}
         \centering
         \begin{subfigure}[b]{0.28\textwidth}
             \centering
             \includegraphics[width=\textwidth]{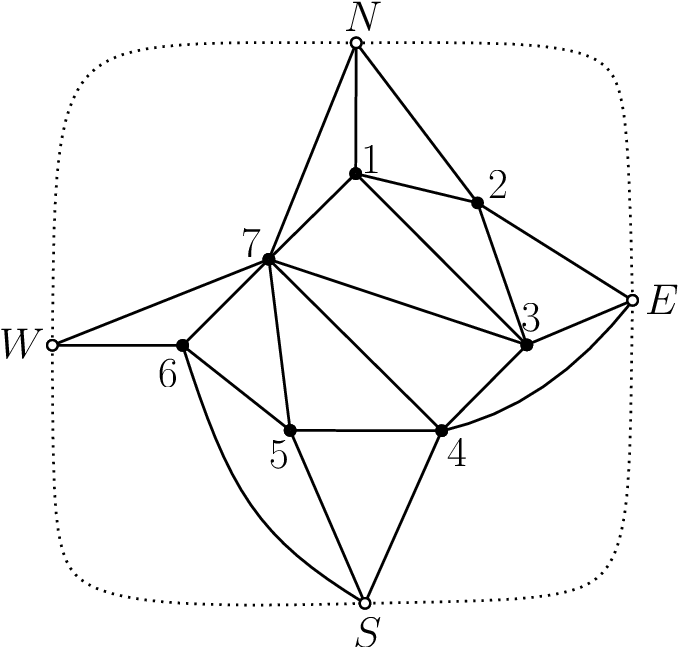}
             \caption{ }
             \label{extended}
         \end{subfigure}
         \hspace{0.5cm}
         \begin{subfigure}[b]{0.3\textwidth}
             \centering
             \includegraphics[width=\textwidth]{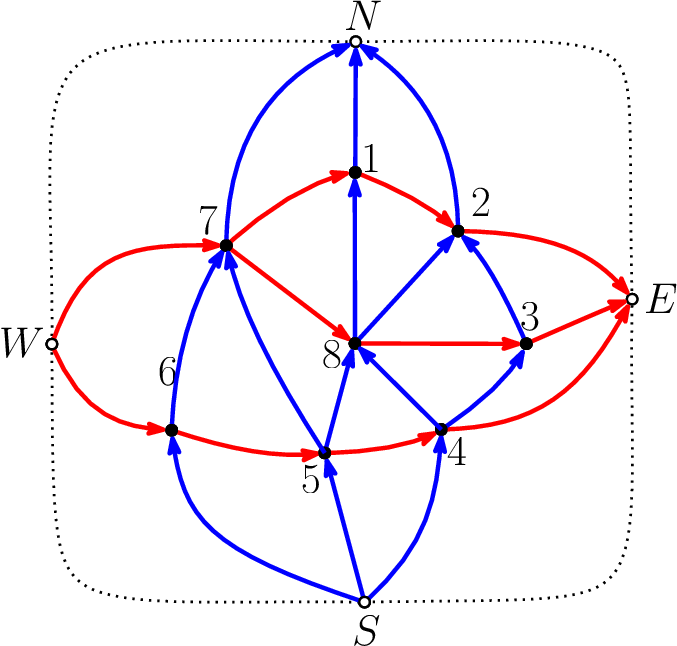}
             \caption{ }
             \label{flip}
         \end{subfigure}
         \hspace{0.5cm}
         \begin{subfigure}[b]{0.3\textwidth}
             \centering
             \includegraphics[width=\textwidth]{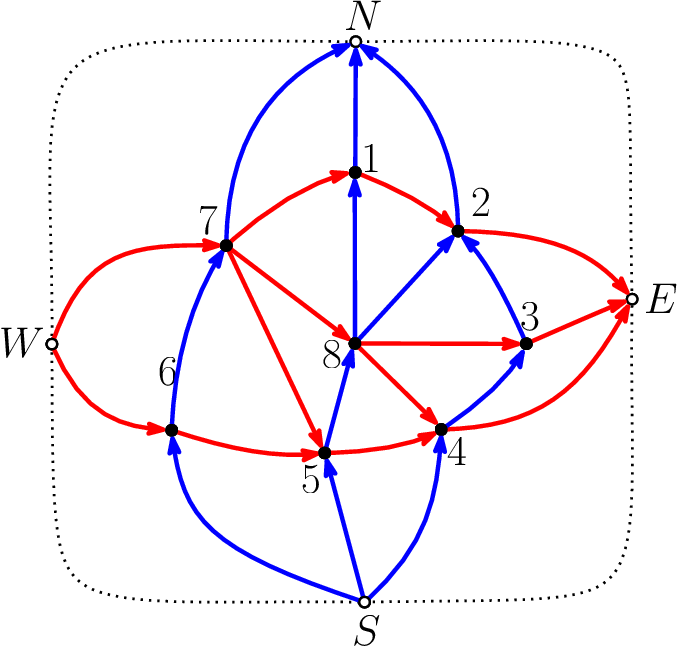}
             \caption{ }
             \label{flip2}
         \end{subfigure}
        
            \caption{(a) An extended outerplanar graph. 
            (b) Presence of flippable edges $(5,7)$ and $(4,8)$. 
            (c) Altered labels of the flippable edges $(5,7)$ and $(4,8)$.}
            \label{Flippablr}
\end{figure}
\end{definition}
\begin{definition}
    Floor plan $(\mathcal{F})$ \cite{rinsma1988existence}: A floor plan (layout) decomposes a polygon into smaller component polygons via straight-line segments. The outer polygon is called boundary of the floor plan, while the smaller component polygons are termed modules. Two modules are adjacent if they share a wall segment; mere point contact (four joint where four line segments meet) does not constitute adjacency. A special class of floor plans is the rectangular floor plan (RFP), in which the boundary and all modules are rectangular.
\end{definition}

\begin{definition}

A rectangular layout $\mathcal{F}$ is said to be \emph{area-universal}~\cite{eppstein2009area} if for every assignment of positive areas to the rectangles of $\mathcal{F}$, there exists a drawing of $\mathcal{F}$ in the plane in which each rectangle realizes its assigned area. In such a drawing, only the sizes of the rectangles may change; the combinatorial structure of $\mathcal{F}$ (i.e., the adjacencies between rectangles) must remain exactly the same as in the original layout. 
\end{definition}

\subsection{Motivation and Contributions}
A significant amount of research has focused on the construction of rectangular floor plans. In 1985, Kozminski and Kinnen \cite{kozminski1985rectangular} showed that a PTPG admits a rectangular floor plan if and only if it can be augmented with four external vertices such that the resulting \emph{extended graph} satisfies the following two properties: (i) every interior face is a triangle and the exterior face is a quadrilateral; and (ii) the extended graph contains no separating triangles. Later, Kant and He~\cite{kant1997regular} introduced the concept of \emph{regular edge labeling} and proposed two linear-time algorithms for computing such labelings. They proved that for any extended graph $E(\mathcal{G})$, a regular edge labeling can be found in linear time, and that the corresponding rectangular floor plan defined by this labeling can also be constructed in linear time. Regular edge labelings have also been investigated by Fusy [\cite{fusy2005transversal}, \cite{fusy2009transversal}], who introduced the term transversal structures for this concept. These labelings are strongly connected to other edge-coloring frameworks on planar graphs, which are useful for representing straight-line drawings and orthogonal polyhedral structures \cite{eppstein2010steinitz}.

A plane graph may admit multiple extended graphs and multiple rectangular layouts with respect to these extended graphs, not all of which are area-universal. Rinsma \cite{rinsma1987nonexistence}
demonstrated explicit examples where specific outerplanar graphs with given area assignments could not be realized by any rectangular layout. Motivated by architectural floor plans, where only a subset of room adjacencies may be specified, Rinsma~\cite{rinsma1988rectangular} considered the following related problem: given a tree $T$, does there exist a rectangular layout $L$ such that $T$ is a spanning tree of the dual graph of $L$? She proved that such a layout always exists. However, the layouts produced by her algorithm are not necessarily area-universal.

To address the open problems posed by Rinsma, Eppstein et al.~\cite{eppstein2012area} established a geometric characterization that is both necessary and sufficient for a rectangular layout to be area-universal (see Theorem~\ref{thm:area-universal-characterization}). In addition, they modified Rinsma’s construction to produce area-universal layouts and proved that for every tree $T$, there exists an area-universal rectangular layout $\mathcal{F}$ such that $T$ is a spanning tree of the dual graph of $\mathcal{F}$.

\noindent In a rectangular layout $\mathcal{F}$, a line segment is defined as a sequence of consecutive inner edges of $\mathcal{F}$. A line segment is said to be maximal if it is not contained within any other line segment. Moreover, a maximal line segment is one-sided if it is a side of at least one rectangle in a rectangular layout.
\begin{theorem}
    \label{thm:area-universal-characterization}
    A rectangular layout $\mathcal{F}$ is said to be area-universal if and only if every internal maximal line segment of $\mathcal{F}$ is one-sided. \cite{eppstein2012area}
\end{theorem}

\noindent The rectangular layout depicted in Figure~\ref{ae} is not area-universal, since the maximal internal line segment $o$ (highlighted in bold) does not align with the side of any rectangle. In contrast, the rectangular layout shown in Figure~\ref{be} is area-universal, as every maximal line segment aligns with the side of some rectangle.
\begin{figure}
         \centering
         \begin{subfigure}[b]{0.18\textwidth}
             \centering
             \includegraphics[width=\textwidth]{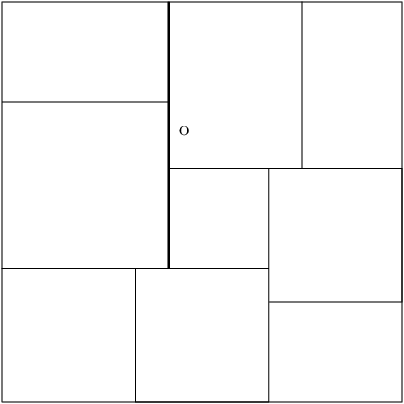}
             \caption{ }
             \label{ae}
         \end{subfigure}
         \hspace{1cm}
         \begin{subfigure}[b]{0.18\textwidth}
             \centering
             \includegraphics[width=\textwidth]{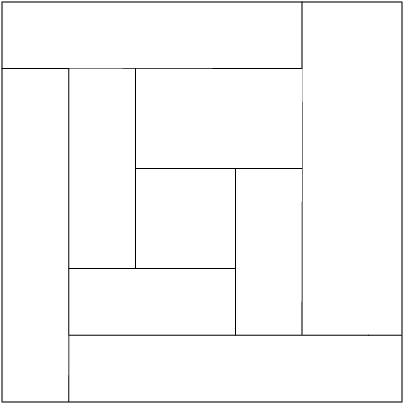}
             \caption{ }
             \label{be}
         \end{subfigure}
        
            \caption{(a) A rectangular layout that is not area-universal. 
            (b) An area-universal rectangular layout.}
            \label{graph1}
\end{figure}

\noindent Eppstein et al. ~\cite{eppstein2012area} proposed an algorithm to construct area-universal rectangular layouts. However, these
algorithms are not fully polynomial, but are fixed-parameter tractable for a parameter
related to the number of separating four-cycles in $\mathcal{G}$. As a result, the problem remains computationally difficult 
in the general case. Later studies have therefore emphasized restricted graph classes. 
For instance, Kumar et al. \cite{kumar2021transformations} 
identified a family of rectangularly dualizable graphs (RDGs) for which every RDG can be realized as an area-universal rectangular layout in polynomial time. 

Outerplanar graphs, which have all their vertices on the boundary of a single face when embedded in the plane, are fascinating due to their structural properties. In this paper, we determine graph-theoretic properties of outerplanar graphs that admit area-universal rectangular layouts. In Section \ref{necessarysection}, we establish necessary and sufficient conditions for biconnected outerplanar graphs to admit such layouts. In Section \ref{Highercase}, we present constructive methods for generating area-universal rectangular layouts and describe several sufficient construction techniques. Furthermore, in Section \ref{lowercase}, we outline an additional construction that enables the generation of all possible area-universal rectangular layouts corresponding to a given biconnected outerplanar graph. Hence, we enumerate all area-universal layouts that can be constructed while satisfying the prescribed adjacency constraints.

\section{Necessary and Sufficient Condition}\label{necessarysection}

We endeavour to establish a comprehensive understanding of outerplanar graphs
that can accommodate area-universal layouts while also identifying those that
cannot. The illustration presented in Figure \ref{outer} showcases the smallest graph known to lack an area-universal layout. It has been established in \cite{eppstein2009area} that a layout qualifies as area-universal if and only if it adheres to the criterion of being one-sided. An area-universal layout is characterized by the property that every maximal line segment serves as a side of at least one rectangle in the layout. Theorem \ref{t1} outlines the essential requirements
needed to determine whether an outerplanar graph can have an area-universal layout, which is also known as a one-sided layout. Eppstein et al. \cite{eppstein2012area} showed that if a regular edge labeling contains a flippable edge, then the corresponding rectangular layout has at least one maximal line segment that is not one-sided. The following lemma will be used throughout the paper.

\begin{lemma}\label{lemmaproof}\cite{eppstein2012area}
Let $\mathcal{G}$ be a plane graph and let $E(\mathcal{G})$ denote its extended graph. Let $r(\mathcal{G})$ be a regular edge labeling of $E(\mathcal{G})$, and let $\mathcal{F}$ be the rectangular floor plan corresponding to $r(\mathcal{G})$. Then $\mathcal{F}$ is area-universal if and only if $r(\mathcal{G})$ contains no flippable edge.
 
\end{lemma}



\begin{figure}[H]
   \centering
    \includegraphics[width=0.17\textwidth]{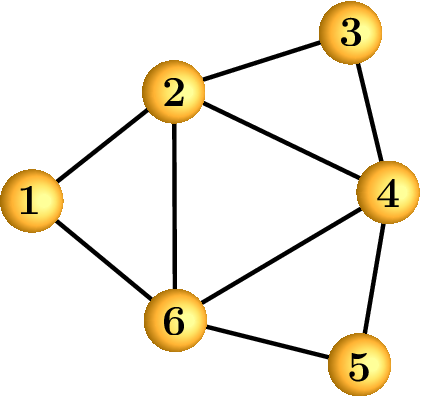}
    \caption{The smallest outerplanar graph G not possessing any area-universal layout}
   \label{outer}
 \end{figure}

\begin{theorem}
    \label{t1}
    
    Let $\mathcal{G}$ be a biconnected outerplanar triangulated graph of order $n>3$. Then there exists at least one extended graph $E(\mathcal{G})$ of $\mathcal{G}$ that admits an area-universal rectangular layout if and only if $\mathcal{G}$ contains exactly two vertices of degree two.

\end{theorem}

\begin{proof}

Considering $\mathcal{G}$ as a biconnected outerplanar triangulated graph implies that the minimum degree of $\mathcal{G}$ is at least two. We define its weak dual $T$ by introducing one vertex for each interior triangular face and joining two such vertices whenever the corresponding faces share an internal edge (see Figure \ref{labelthm}). It follows that $T$ is a tree. Each leaf of $T$ corresponds to a triangular face of $\mathcal{G}$ that shares exactly one internal edge, where the opposite boundary vertex in $\mathcal{G}$ has degree two. Thus, the degree-2 vertices of $\mathcal{G}$ are in one-to-one correspondence with the leaves of $T$. Moreover, since every tree with $|V(T)|>1$ has at least two leaves, $\mathcal{G}$ must contain at least two degree-2 vertices whenever $n>3$.

First, we assume that $\mathcal{G}$ has exactly two degrees-2 vertices. We aim to demonstrate the existence of an \emph{extended graph} that admits an area-universal rectangular layout. To this end, we present an algorithm (see Algorithm \ref{algo2}) that augments $\mathcal{G}$ to obtain an extended graph $E(\mathcal{G})$ such that every rectangular layout corresponding to any regular edge labeling of $E(\mathcal{G})$ is area-universal. By proving that this algorithm consistently generates an area-universal layout for any biconnected outerplanar graph having two degree-2 vertices, we establish the validity of this statement.

Conversely, assuming that an area-universal layout exists corresponding to an extended graph $E(\mathcal{G})$, then we need to show that the proper graph $\mathcal{G}$ has exactly two degree-2 vertices. Here, we will prove the contrapositive statement. For that, we assume that $\mathcal{G}$ has at least three (say $k (k > 2)$) vertices of degree two, then its weak dual $T$ has at least three leaves and therefore, $T$ must contain at least one vertex of degree three. In $T$, a vertex of degree three corresponds to a interior triangular face (a triangular face $\textit{F} = \{v_1, e_1, v_2, e_2, v_3, e_3, v_1\}$, each of whose three edges is shared with another triangular face) in $\mathcal{G}$ (see Figure \ref{labelthm}). Now consider a biconnected outerplanar graph $\mathcal{G}$ that contains an interior triangular face $(v_1, v_2, v_3)$, which corresponds to a degree-three vertex in the weak dual of $\mathcal{G}$. Let $E(\mathcal{G})$ be any extended graph of $\mathcal{G}$. Constructing $E(\mathcal{G})$ requires augmenting $\mathcal{G}$ by adding four cardinal vertices. This augmentation induces a partition of the outer boundary of $\mathcal{G}$ into four contiguous vertex paths $(P_1, P_2, P_3, P_4)$ that appear in clockwise order along the outer face. Each cardinal vertex is adjacent to all vertices in its corresponding path $P_i$, for $1 \leq i \leq 4$. The augmentation is performed such that the resulting graph remains triangulated and contains no separating triangles. For any partition of the outer boundary of $\mathcal{G}$ into four contiguous paths in the augmentation process, if a degree-2 vertex lies on exactly one boundary path, consequently it create a separating triangle. Thus, no single path can contain a degree-2 vertex, each degree-2 vertex is necessarily shared by at least two boundary paths. It follows that the vertices $v_1, v_2, v_3$ of the triangular face must lie on different paths, and consequently, the associated modules are positioned 
along different cardinal directions $(N, W, S, E)$. Without loss of generality, let $v_1, v_2$, and $v_3$ lie on the boundary paths $P_1, P_2$, and $P_3$, respectively. Consequently, the modules associated with $v_1, v_2$, and $v_3$ are positioned along the $N$, $E$, and $S$ cardinal directions, respectively. Since $v_1, v_2$, and $v_3$ are pairwise adjacent and lie on three distinct boundary paths, two of the corresponding modules appear collinear in the rectangular layout, while the third module is orthogonal to these two. In the REL, there are precisely three possible labeling patterns among $v_1, v_2$, and $v_3$ that realize this configuration (see Figure~\ref{labelthm}).

\begin{itemize}
    \item[(i)] $(v_1, v_2), (v_3, v_2) \in T_2$ and $(v_3, v_1) \in T_1$.
    \item[(ii)] $(v_3, v_1), (v_2, v_1) \in T_1$ and $(v_3, v_2) \in T_2$.
    \item[(iii)] $(v_3, v_1), (v_3, v_2) \in T_1$ and $(v_1, v_2) \in T_2$.
\end{itemize}
\begin{figure}
    \centering
    \includegraphics[width=0.6\linewidth]{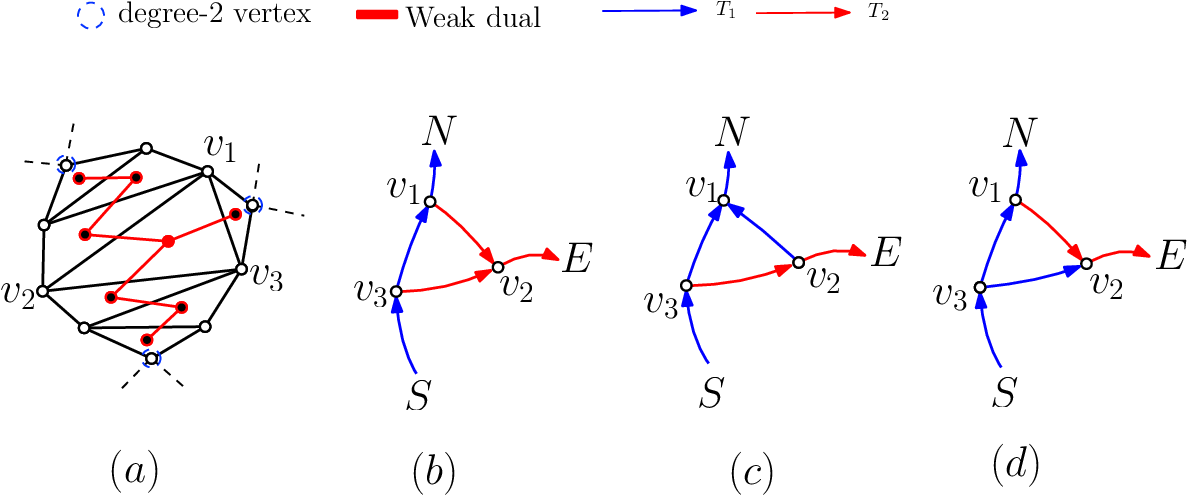}
    \caption{(a) Paths are separated by degree-2 vertices. (b-d) Three possible labeling patterns among $v_1, v_2,$ and $v_3$.}
    \label{labelthm}
\end{figure}
Let $\alpha_1 (\neq v_3)$ be a common neighbour of the vertices $v_1$ and $v_2$. 
Depending on the adjacency in the graph $\mathcal{G}$, either $\alpha_1$ lies on the path $P_1$ or on the path $P_2$. 
In both cases, the edge $(v_1,v_2)$ is the diagonal of the alternately labeled four-cycle $(v_2 - \alpha_1 - v_1 - v_3)$ in the REL, and hence $(v_1,v_2)$ is a flippable edge (see Figure~\ref{label2}). Moreover, in this configuration the edge $(v_3,v_2)$ is also flippable. Let $\alpha_2 (\neq v_1)$ be a common neighbour of the vertices $v_3$ and $v_2$. Depending on the adjacency in the graph, either $\alpha_2$ lies on the path $P_2$ or on the path $P_3$. In both scenarios, the edge $(v_3,v_2)$ is the 
diagonal of the alternately labeled four-cycle $(v_3 - \alpha_2 - v_2 - v_1)$, and is therefore flippable (see Figure~\ref{label2}).

Case~(ii) can be derived directly from Case~(i). In Case~(i), the edge $(v_1,v_2)$ belongs to $T_2$; after flipping its label so that $(v_2,v_1) \in T_1$, we obtain the configuration of Case~(ii), in which $(v_2,v_1)$ remains a flippable edge. Similarly, Case~(iii) can be derived from Case~(i) by starting with $(v_3,v_2) \in T_2$ and flipping its label so that $(v_3,v_2) \in T_1$, yielding the configuration of Case~(iii), where $(v_3,v_2)$ is still a flippable edge.

Therefore, there always exists a flippable edge among $(v_1,v_2)$, $(v_1,v_3)$, and $(v_2,v_3)$, where 
$v_i$, $1 \le i \le 3$, are the vertices of the triangular face corresponding to a degree-3 vertex of the weak dual, depending on the \emph{extended graph} $E(\mathcal{G})$ and the adjacencies in the outerplanar graph $\mathcal{G}$.

 \hfill $\square$

\begin{figure}
    \centering
    \includegraphics[width=0.8\linewidth]{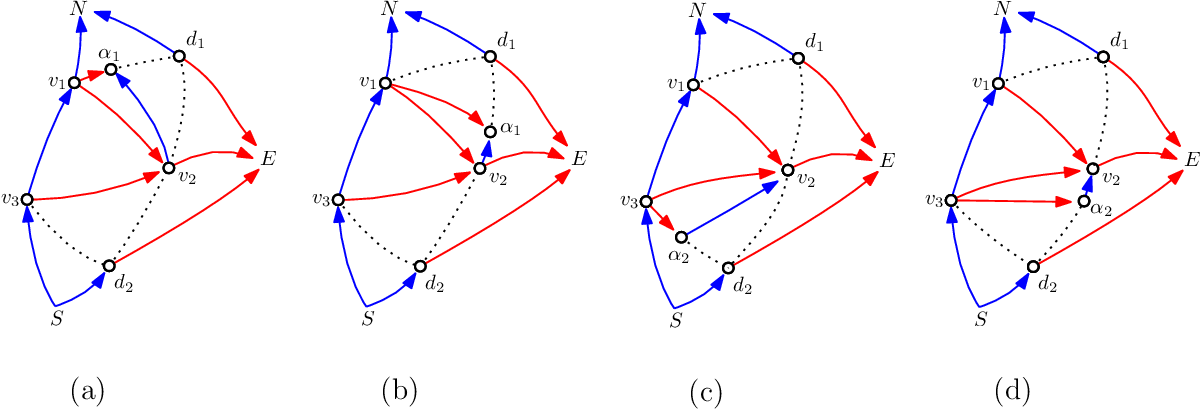}
    \caption{(a-b) edge $(v_1,v_2)$ is flippable, (c-d) edge $(v_3,v_2)$ is flippable.}
    \label{label2}
\end{figure}
\end{proof}

\textbf{Remark:} \noindent An outerplanar graph with exactly two vertices of degree two admits at least one extended graph that corresponds to an area-universal rectangular layout. However, not every extension of such an outerplanar graph yields an area-universal layout. As illustrated in Figure~\ref{Different4comp}, the same outerplanar graph $\mathcal{G}$ can be augmented in two different ways to obtain extended graphs $E(\mathcal{G})$. One of these augmentations produces an area-universal layout (see Figure~\ref{AU}), whereas the other does not (see Figure~\ref{NAU}). This observation highlights the importance of carefully selecting the augmentation. Consequently, it is essential to characterize structural properties of the \emph{extended} outerplanar graph. Theorem~\ref{necessary} identifies the specific types of augmentations that must be avoided in order to obtain an extended graph admitting an area-universal rectangular layout.

\begin{figure}
         \centering
         \begin{subfigure}[b]{0.2\textwidth}
             \centering
             \includegraphics[width=\textwidth]{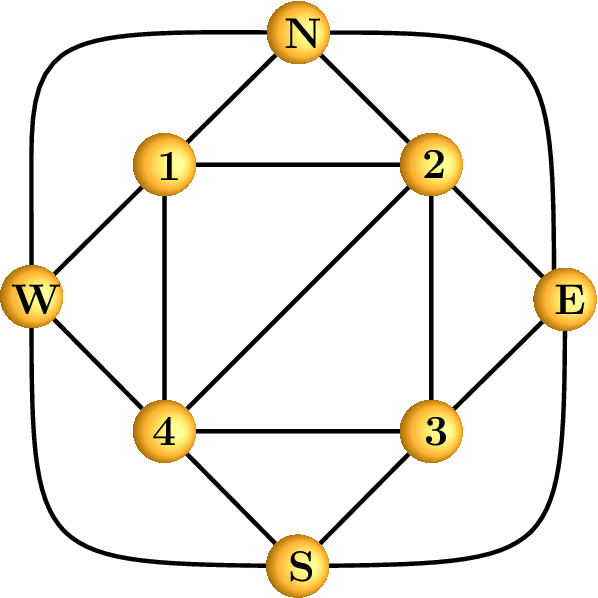}
             \caption{ }
             \label{NAU}
         \end{subfigure}
         \hspace{1cm}
         \begin{subfigure}[b]{0.2\textwidth}
             \centering
             \includegraphics[width=\textwidth]{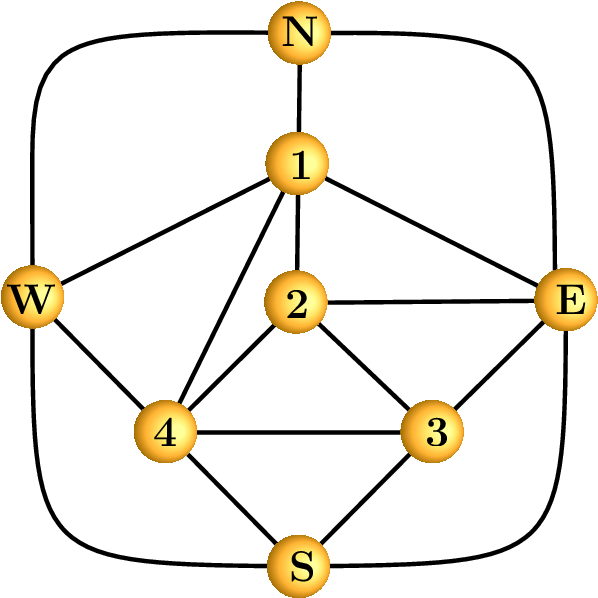}
             \caption{ }
             \label{AU}
         \end{subfigure}
        
            \caption{(a) An extended outerplanar graph that does not admit an area-universal layout. 
            (b) An extended outerplanar graph that admits an area-universal layout.}
            \label{Different4comp}
\end{figure}

\begin{theorem}\label{necessary}
Let $\mathcal{G}$ be an outerplanar graph having exactly two vertices of degree two, and $E(\mathcal{G})$ be its extended graph. If $E(\mathcal{G})$ admits area-universal rectangular layout, then a vertex $v \in V(\mathcal{G})$ must be adjacent to exactly three cardinal vertices in $E(\mathcal{G})$. 
\end{theorem}

\begin{proof}
    Let $\mathcal{G}$ be a biconnected outerplanar graph of order $n>3$. In order to obtain \emph{extended} outerplanar by augmentation, it is not possible for all four cardinal vertices to be adjacent to a single vertex of the outerplanar graph. Likewise, it is also impossible for each vertex of the outerplanar graph to be adjacent only to one cardinal vertex in the \emph{extended} graph (due to triangulation in \emph{extended} outerplanar graph). Assuming that no vertex is adjacent to exactly three cardinal vertices, there must then exist precisely four vertices, denoted by $\mathcal{C}=\{v_{NE}, v_{SE}, v_{WS}, v_{NW}\}$, each of which is adjacent to exactly two cardinal vertices in the extended graph $E(\mathcal{G})$.
    Without loss of generality, assume that $v_{NE}$ lies in the North-East corner (that is, adjacent to $N$ and $E$), $v_{SE}$ in the South-East corner, $v_{WS}$ in the West-South corner, and $v_{NW}$ in the North-West corner (see Figure~\ref{thpr1}). If the edges $(v_{NE}, v_{SE})$, $(v_{SE}, v_{WS})$, $(v_{WS}, v_{NW})$, and $(v_{NW}, v_{NE})$ exist, then this configuration represents the simplest form that does not admit an area-universal layout (see Figure~\ref{NAU}). 

    Since $\mathcal{G}$ is outerplanar, each vertex must be adjacent to at least one cardinal vertex in the \emph{extended} outerplanar graph $E(\mathcal{G})$. Consequently, all vertices in $V(\mathcal{G}) \setminus \{v_{NE}, v_{SE}, v_{WS}, v_{NW}\}$ are adjacent to exactly one cardinal vertex. We define a path in $\mathcal{G}$ to be \emph{non-trivial} if it is of length at least three and its endpoints lie in $\mathcal{C}$ in the prescribed order. We now examine the conditions under which there exists at least one non-trivial path among the vertices of $\mathcal{C}$.

    \begin{itemize}
        \item[1).] Suppose there exists exactly one non-trivial path in $\mathcal{G}$. Without loss of generality, let $p_1$ be a non-trivial path from $v_{NE}$ to $v_{SE}$, denoted by 
                $p_1 = \{u_{1}, u_{2}, \allowbreak \ldots, u_{k}\},
                $
        where $u_{1}=v_{NE}$ and $u_{k}=v_{SE}$. In this case, all the edges $(u_{j}, u_{i}),\  \forall j>i$ are labeled as $T_1$ (likewise if end points of a non-trivial path are $v_{NW}$ and $v_{NE}$, then edges labeled as $T_2$) in each REL. Furthermore, there must exist a vertex $u_c \in p_1$ that is adjacent to both $v_{NW}$ and $v_{WS}$. Consequently, either the set $\{u_c, u_{c-1}, v_{NW}, v_{WS}\}$ or the set $\{u_c, u_{c+1}, v_{WS}, v_{NW}\}$ (see Figures~\ref{thpr2} and~\ref{thpr3}) labeled alternately in the REL. Hence, a flippable ($(u_c, v_{NW})\text{ or }(u_c, v_{WS}) $) edge necessarily exists in this configuration.
\begin{figure}
         \centering
         \begin{subfigure}[b]{0.23\textwidth}
             \centering
             \includegraphics[width=\textwidth]{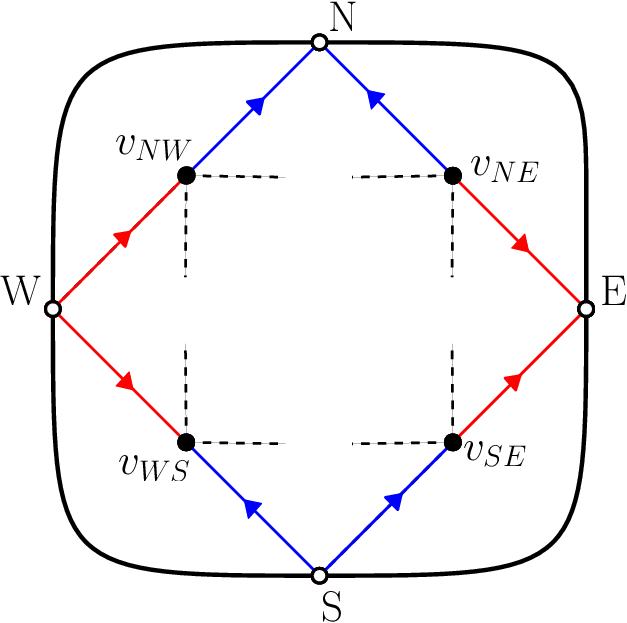}
             \caption{ }
             \label{thpr1}
         \end{subfigure}
         \hspace{1cm}
         \begin{subfigure}[b]{0.23\textwidth}
             \centering
             \includegraphics[width=\textwidth]{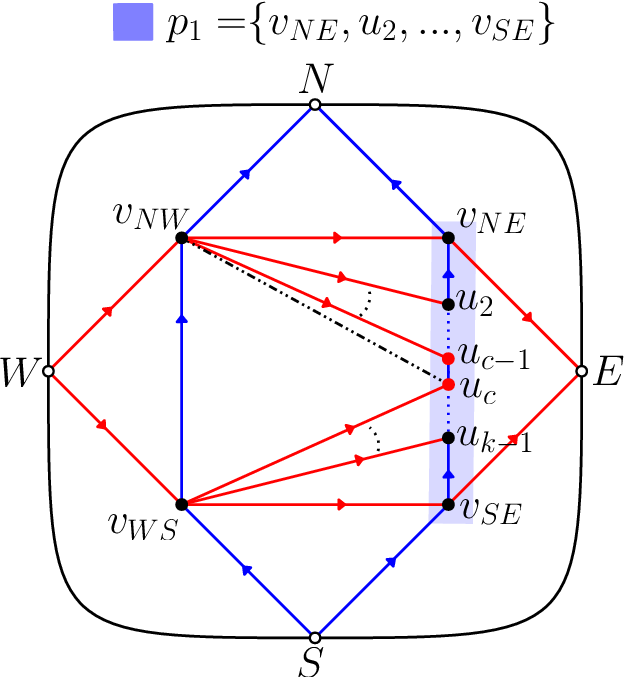}
             \caption{ }
             \label{thpr2}
         \end{subfigure}
         \hspace{1cm}
         \begin{subfigure}[b]{0.23\textwidth}
             \centering
             \includegraphics[width=\textwidth]{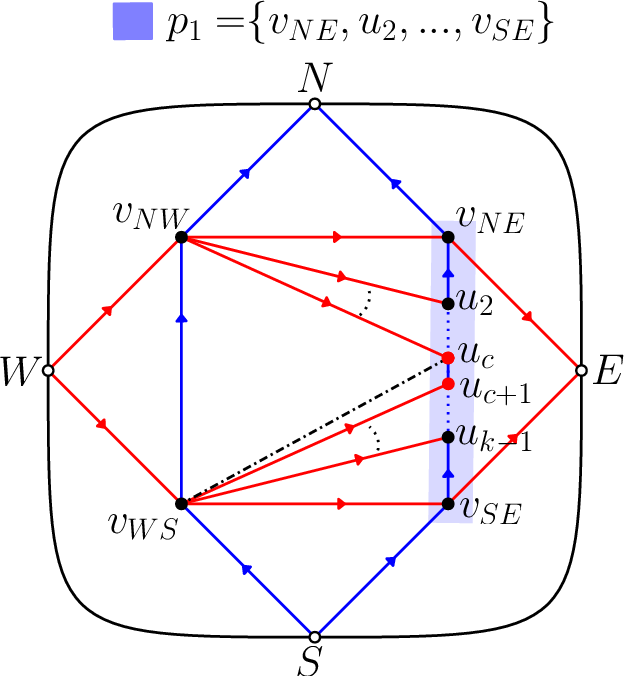}
             \caption{ }
             \label{thpr3}
         \end{subfigure}
        
            \caption{(a) Adjacency of vertices $v_{NE}, v_{ES},v_{WS},$ and $v_{NW}$ with cardinal vertices. (b) Existence of flipable edge $(v_{NW}\ u_c$). (c) Existence of flipable edge $(u_c\ v_{WS})$,}
            \label{Different4comp1}
\end{figure}

        \item[2).] If there exist exactly two non-trivial paths, say $p_1$ and $p_2$, then the analysis can be divided into two subcases: (i) when the two paths are consecutive, and (ii) when the two paths are non-consecutive.

        \begin{itemize}
            \item[2.1).] Let $p_1$ and $p_2$ be two consecutive non-trivial paths. Without loss of generality, let 
            $ p_1=\{u_{11}, u_{12}, \allowbreak \ldots, u_{1k_1}\}$, $p_2=\{ u_{21}, u_{22}, \ldots, u_{2k_2}\}$, where $u_{11}= v_{NE}, u_{1k_1}=u_{2k_2}=v_{SE},$ and $u_{21}= v_{WS}$. This configuration can be further subdivided into three distinct cases.

            \begin{itemize}
                \item[i] The adjacencies of the graph are such that $\deg(v_{WS})=\deg(v_{SE})=\deg(v_{NE})=2$ and $\deg(v_{NW})>2$. By Theorem~\ref{t1}, it follows that $\mathcal{G}$ does not admit an area-universal layout.

                \item[ii]Either there exists a vertex $u_{1c} \in p_1$ that is adjacent to both $v_{NW}$ and $v_{WS}$, or there exists a vertex $u_{2c} \in p_2$ that is adjacent to both $v_{NW}$ and $v_{NE}$. This situation can be resolved by an argument analogous to that of Case~1 (see Figures \ref{thpr4} and \ref{thpr4_1}).
                \item[iii] When $\deg(v_{NW}) = n-1$, it is straightforward to verify that the edge $(v_{NW}, v_{SE})$ is a flippable edge (see Figure \ref{thpr5}).
             \end{itemize}
\begin{figure}
     \centering
     \begin{subfigure}[b]{0.23\textwidth}
         \centering
         \includegraphics[width=\textwidth]{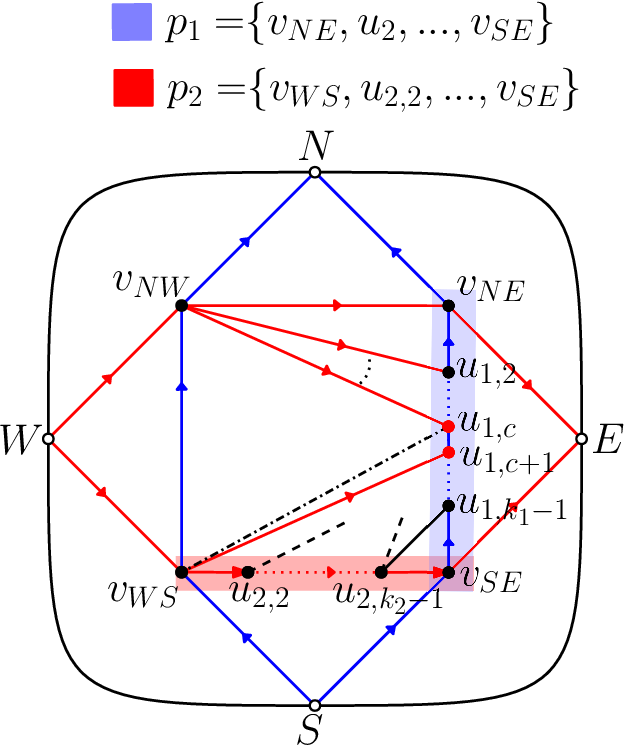}
         \caption{ }
         \label{thpr4}
     \end{subfigure}
     \hspace{1cm}
      \begin{subfigure}[b]{0.23\textwidth}
         \centering
         \includegraphics[width=\textwidth]{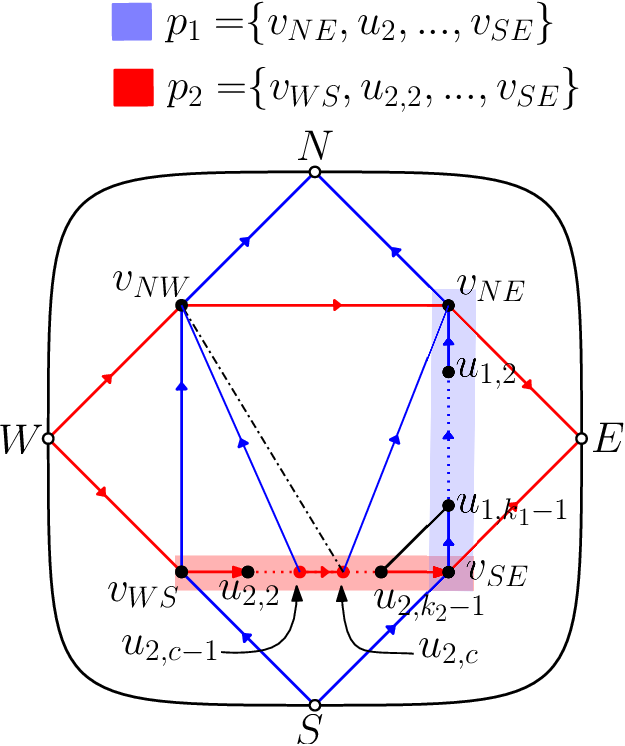}
         \caption{ }
         \label{thpr4_1}
     \end{subfigure}
     \hspace{1cm}
     \begin{subfigure}[b]{0.23\textwidth}
         \centering
         \includegraphics[width=\textwidth]{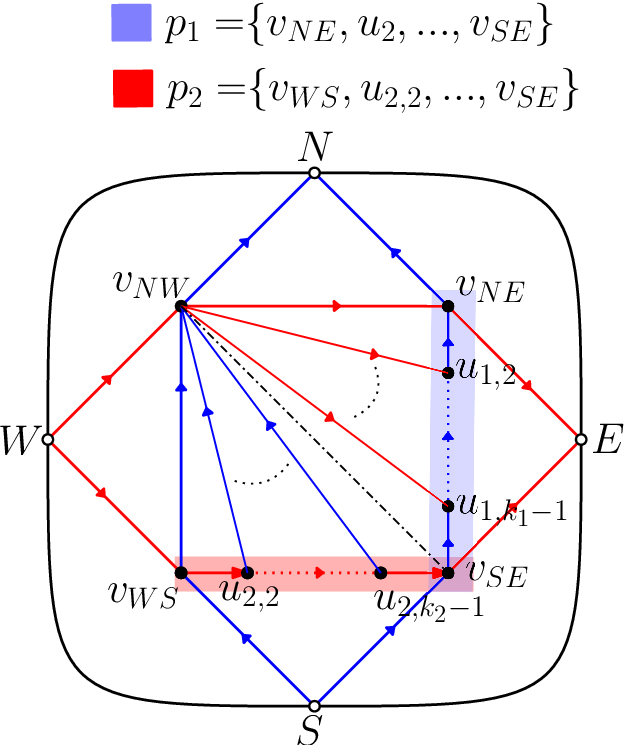}
         \caption{ }
         \label{thpr5}
\end{subfigure}

        \caption{(a) Existence of a flipable edge ($v_{WS}\ u_{1,c}$). (b) Existence of a flipable edge ($v_{NW}\ u_{2,c}$). (c) Existence of a flipable edge ($v_{NW}\ v_{SE}$)}
        \label{Different4comp5}
\end{figure}

            \item[2.2).]Let $p_1$ and $p_2$ be two non-consecutive non-trivial paths. Without loss of generality, let $p_1=\{u_{11}, u_{12}, \allowbreak \ldots, u_{1\ k_1-1}, u_{1k_1}\}$, where $u_{11} = v_{NE},\ u_{1k_1}= v_{SE}$ and $p_2=\{u_{21}, u_{22},\allowbreak \ldots, u_{2\ k_2-1}, u_{2k_2}\}$, where $u_{21}=v_{NW},\ u_{2k_2}= v_{WS}$.
            Note that $v_{NW}$ and $v_{NE}$ cannot both have degree $2$, and likewise $v_{WS}$ and $v_{SE}$ cannot both have degree $2$. For $x\in\{NW,NE,SE,WS\}$ with $\deg(v_x)>2$ we associate an index $m\in\{1,2\}$ by
            \[
            m=\begin{cases}
            1, & x\in\{NW,WS\},\\[2pt]
            2, & x\in\{NE,SE\},
            \end{cases}
            \]
           Let $p_m=\{u_{m1},u_{m2},\dots,u_{m k_m}\}$ be the path whose vertices may be adjacent to $v_x$, listed in the natural order along the path. We call the vertex $u_{m,n}\in p_m$ the \emph{last adjacent vertex} of $v_x$ on $p_m$ if $(v_x,u_{m,n})\in E(\mathcal{G})$ and if $n>1$ then $(v_x,u_{m,n-1})\notin E(\mathcal{G})$. ; equivalently, $(v_x,u_{m,n})\in E(\mathcal{G})$ and, if $n<k_m$, then $(v_x,u_{m,n+1})\notin E(\mathcal{G})$. If $u$ is any neighbour of $v_x$ on the opposite path of $p_m$, then either the edge $(u_{m,n},u)$ or the edge $(u_{m,n},v_x)$  is a flippable edge (see Figures \ref{thpr6} and \ref{thpr7}).

        \end{itemize}
\begin{figure}
     \centering
     \begin{subfigure}[b]{0.27\textwidth}
         \centering
         \includegraphics[width=\textwidth]{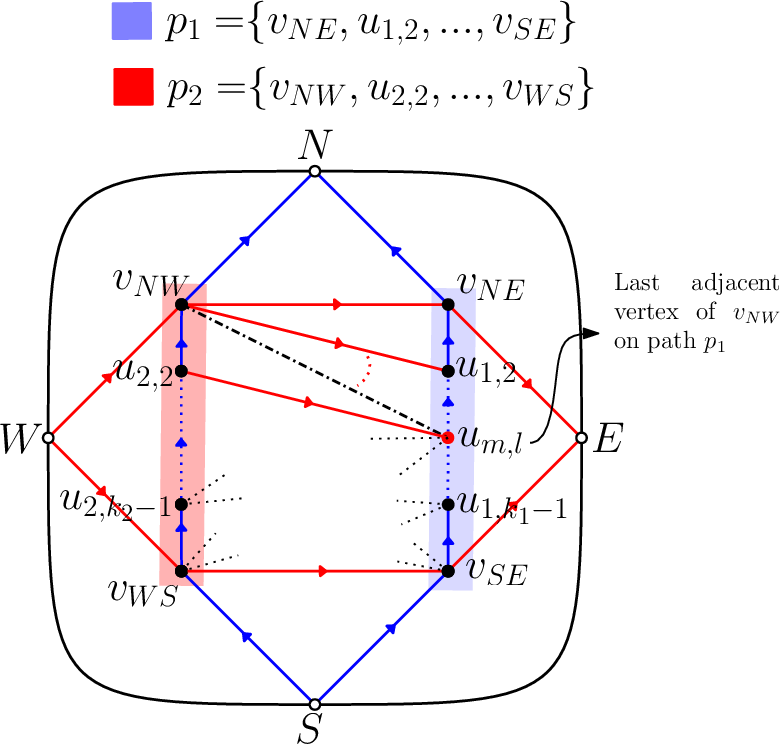}
         \caption{ }
         \label{thpr6}
     \end{subfigure}
     \hspace{0.5cm}
     \begin{subfigure}[b]{0.27\textwidth}
         \centering
         \includegraphics[width=\textwidth]{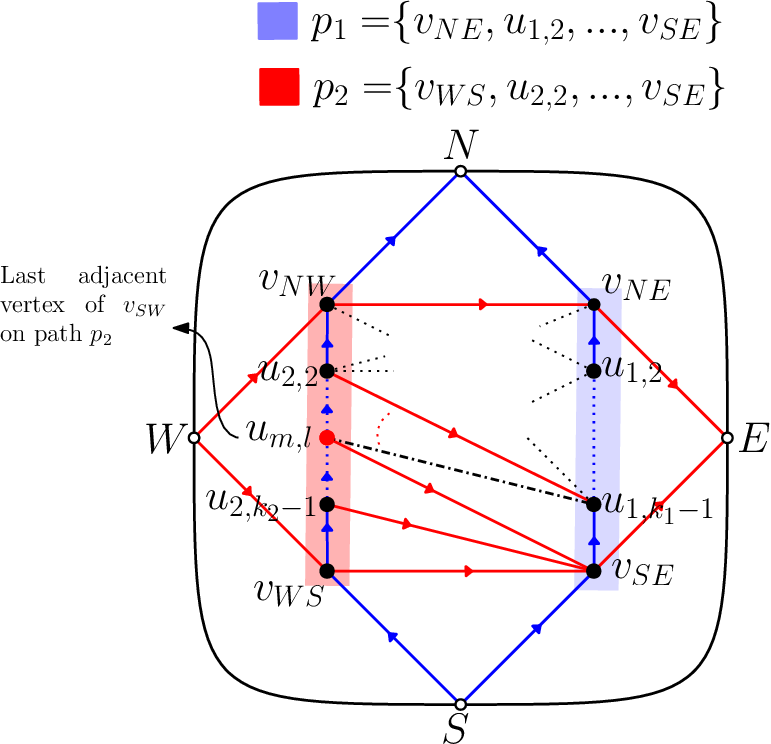}
         \caption{ }
         \label{thpr7}
     \end{subfigure}
     \hspace{0.5cm}
     \begin{subfigure}[b]{0.22\textwidth}
         \centering
         \includegraphics[width=\textwidth]{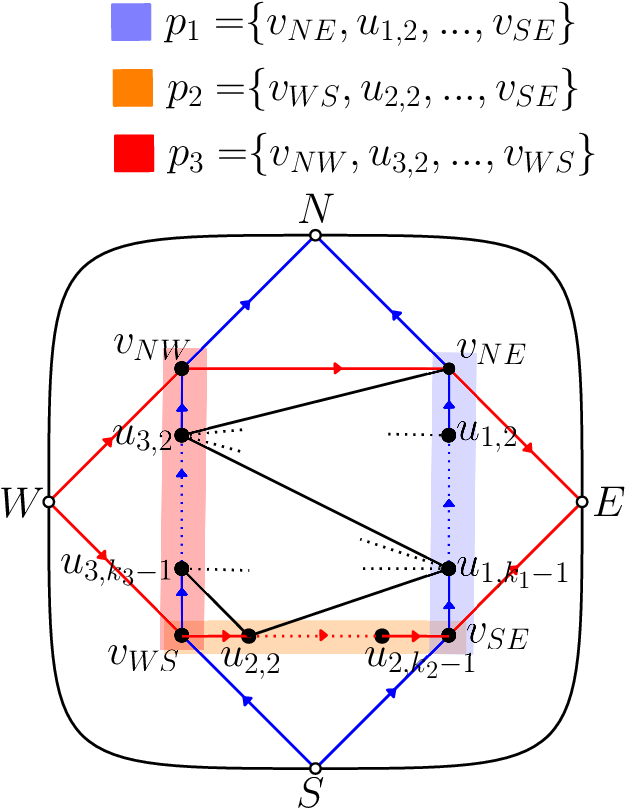}
         \caption{ }
         \label{thpr8}
     \end{subfigure}
        \caption{(a) Existence of a flipable edge ($v_{x}\ u_{m,l}$), $x=NW,\ m=1$,  (b) Existence of a flipable edge ($u_{m,l}\ u$), $m=2$ and $u=u_{1,k_1-1}$ (c) $deg(v_{NW})=deg(v_{SE})=deg(v_{WS})=2$ in the outerplanar graph $\mathcal{G}$.}
        \label{Different4comp2}
\end{figure}

        \item[3).]If there exist exactly three non-trivial paths, say $p_1, p_2,$ and $p_3$, we may assume without loss of generality that  
        $
        p_1 = \{u_{11}, u_{12}, \ldots, u_{1k_1}\},\  u_{11}=v_{NE}, \ u_{1k_1}=v_{SE},
        $  
        $
        p_2 = \{u_{21}, u_{22}, \ldots, u_{2k_2}\}, \ u_{21}=v_{WS}, \ u_{2k_2}=v_{ES},
        $  
        and  
        $
        p_3 = \{u_{31}, u_{32}, \ldots, u_{3k_3}\}, \ u_{31}=v_{NW}, \ u_{3k_3}=v_{WS}.
        $  
        In such a configuration, it is clear that vertices $v_{NW}$ and $v_{NE}$ cannot both have degree two. Furthermore, if no vertex on path $p_1$ is adjacent to $v_{WS}$, no vertex on path $p_3$ is adjacent to $v_{ES}$, and no vertex on path $p_1$ is adjacent to either $v_{NW}$ or $v_{NE}$, then the resulting outerplanar graph contains three vertices of degree two (see Figure \ref{thpr8}). Such a graph does not admit an area-universal layout (see Theorem~\ref{t1}). In all other admissible adjacency configurations, exactly two vertices will have degree two, which we analyze further in the following two subcases.  
        \begin{itemize}
            \item[i] When $\deg(v_{NW})$ and $\deg(v_{NE})$ are both greater than two, this situation arises only if there exists a vertex $u_{2c}$ on the path $p_2$ such that it is adjacent to both $v_{NE}$ and $v_{NW}$. In this case, we obtain a flippable edge—either $(u_{2c}, v_{NW})$ or $(u_{2c}, v_{NE})$—as already discussed in Case~1 (see Figure \ref{thpr9}).
            \item[ii]  When $\deg(v_x) > 2$ for some $x \in \{NW, NE\}$, define $m=3$ if $x=NE$ and $m=1$ if $x=NW$. Let $u_{m,l}$ denote the last vertex on path $p_m$ adjacent to $v_x$; that is,  
            $
            (v_x, u_{m,l}) \in E(\mathcal{G}), \  \text{and if } l < k_m \text{ then } (v_x, u_{m,l+1}) \notin E(\mathcal{G}).
            $
            In this situation, either the edge $(v_x, u_{m,l})$ or the edge $(u_{m,l}, u_{p,2})$ must be flippable, where $p=2$ if $x=NE$ and $p=3$ if $x=NW$ (see Figures \ref{thpr10} and \ref{thpr11}).

        \end{itemize}
        \begin{figure}
     \centering
     \begin{subfigure}[b]{0.26\textwidth}
         \centering
         \includegraphics[width=\textwidth]{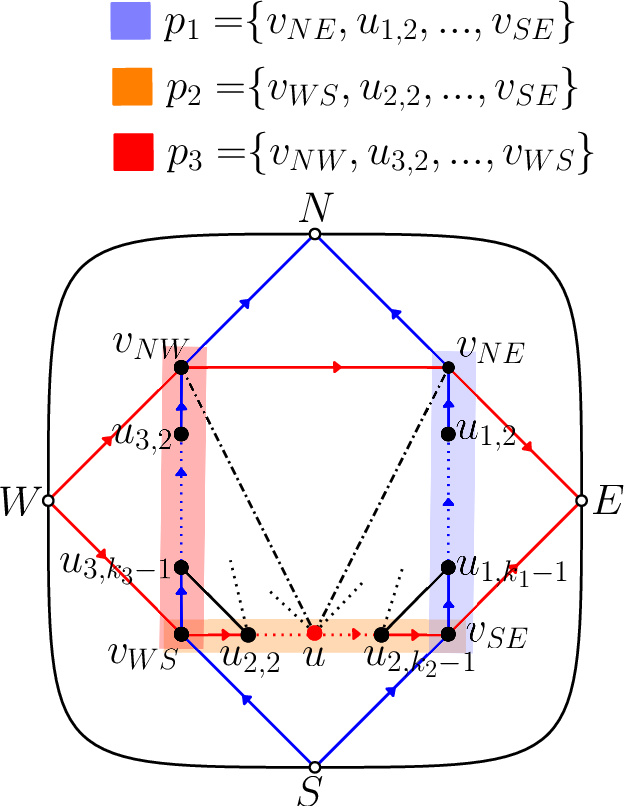}
         \caption{ }
         \label{thpr9}
     \end{subfigure}
     \hspace{0.5cm}
     \begin{subfigure}[b]{0.33\textwidth}
         \centering
         \includegraphics[width=\textwidth]{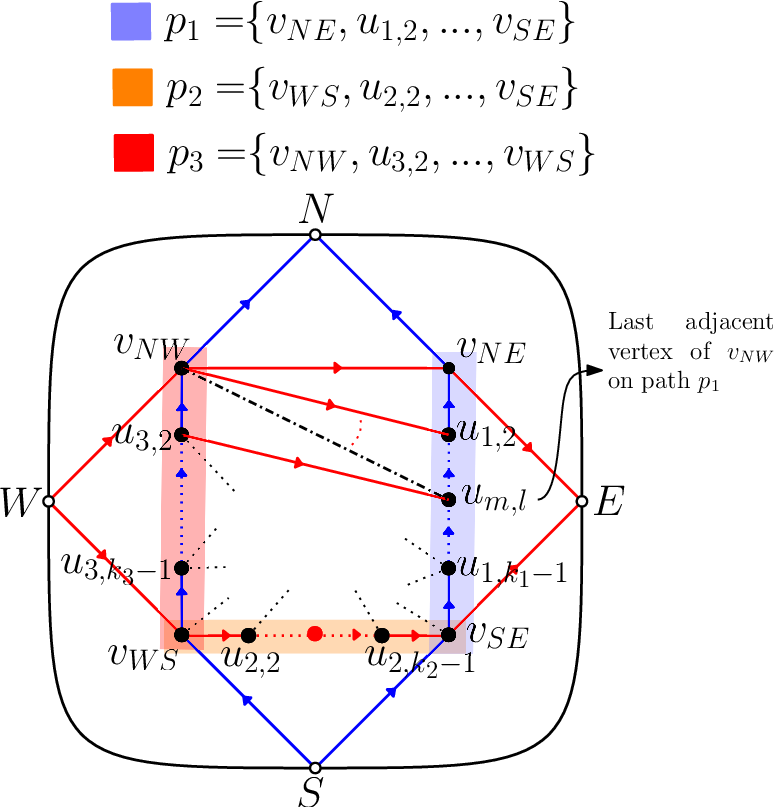}
         \caption{ }
         \label{thpr10}
     \end{subfigure}
     \hspace{0.5cm}
     \begin{subfigure}[b]{0.33\textwidth}
         \centering
         \includegraphics[width=\textwidth]{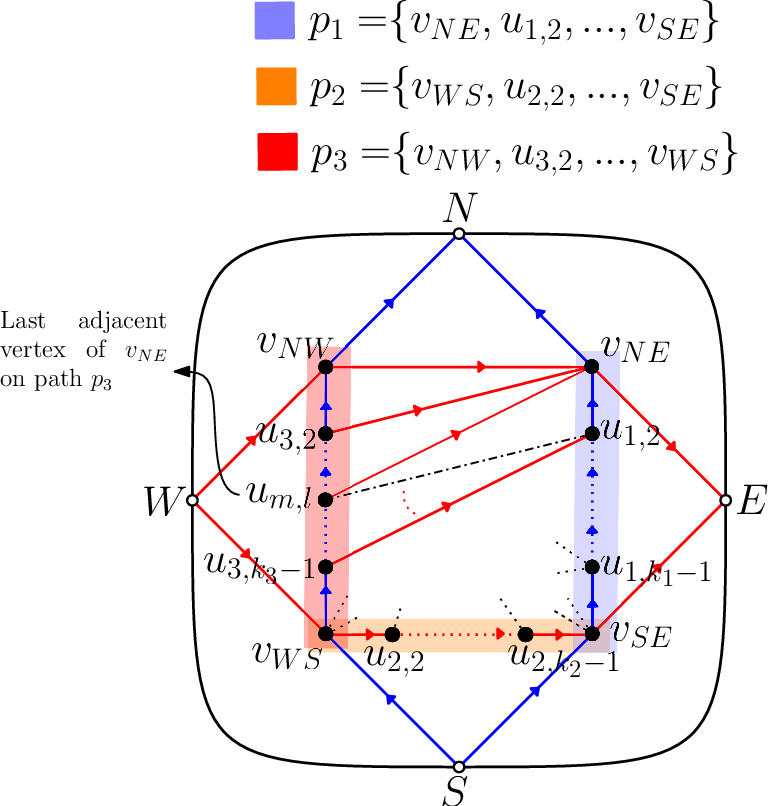}
         \caption{ }
         \label{thpr11}
     \end{subfigure}
        \caption{(a) Existence of a flipable edge either ($v_{NW}\ u$) or ($v_{NE}\ u$),  (b) Existence of a flipable edge ($u_{m,l}\ v_{x}$), $m=1$ and $x=NW$, (c) Existence of a flipable edge ($u_{m,l}\ u_{p,2}$), $m=3$ and $p=1$.}
        \label{Different4comp3}
\end{figure}
        \item[4).] If there exist exactly four non-trivial paths, say $p_1, p_2, p_3$ and $p_4$, we may assume without loss of generality that  
        $
        p_1 = \{u_{11}, u_{12}, \ldots, u_{1k_1}\},\  u_{11}=v_{NE}, \ u_{1k_1}=v_{SE},
        $  
        $
        p_2 = \{u_{21}, u_{22}, \ldots, u_{2k_2}\}, \ u_{21}=v_{WS}, \ u_{2k_2}=v_{ES},
        $    
        $
        p_3 = \{u_{31}, u_{32}, \ldots, u_{3k_3}\}, \ u_{31}=v_{NW}, \ u_{3k_3}=v_{WS},
        $
        and
        $
        p_4 = \{u_{41}, u_{42}, \ldots, u_{4k_4}\}, \allowbreak \ u_{41}=v_{NW}, \ u_{4k_4}=v_{NE}.
        $ We exclude the configurations that contain three or more degree-2 vertices in the outerplanar graph $\mathcal{G}$ (these are known not to admit an area-universal layout (Theorem \ref{t1}). For any consecutive paths $p_m,p_n$ (with $m,n\in\{1,2,3,4\}$ and $n\equiv m+1\pmod 4$), suppose $v_x$ with $x\in\{NW,NE,SE,WS\}$ does not lie on both $p_m$ and $p_n$.  Define $u_{j,k}$ to be the vertex in $(p_m\cup p_n)\cap N(v_x)$ such that if $k>2$ then $(v_x,u_{j,k-1})\notin E(\mathcal{G})$ or if $k<k_j$, then $(v_x,u_{j,k+1})\notin E(\mathcal{G})$. (i.e. the last vertex on $p_m$ or $p_n$ that is adjacent to $v_x$).  Let $p_{m+2}$ denote the path opposite $p_m$ (indices modulo $4$).  If $u_{j,k}$ lies on $p_m$ (respectively $p_n$), let $u$ be the neighbor of $v_x$ that lies on the opposite path $p_{m+2}$. Then either $(u_{j,k}, u)$ or $(u_{j,k}, v_x)$ must be a flippable edge (see Figures \ref{thpr12} and \ref{thpr13}).
        \end{itemize}

     \noindent In all cases considered above, if no vertex of $\mathcal{G}$ is adjacent to exactly three cardinal vertices in the extended graph $E(\mathcal{G})$, then every regular edge labeling of $E(\mathcal{G})$ contains at least one flippable edge. Consequently, the corresponding rectangular layout cannot be area-universal (see Lemma \ref{lemmaproof}).
 \hfill       $\square$

\begin{figure}
     \centering
     \begin{subfigure}[b]{0.3\textwidth}
         \centering
         \includegraphics[width=\textwidth]{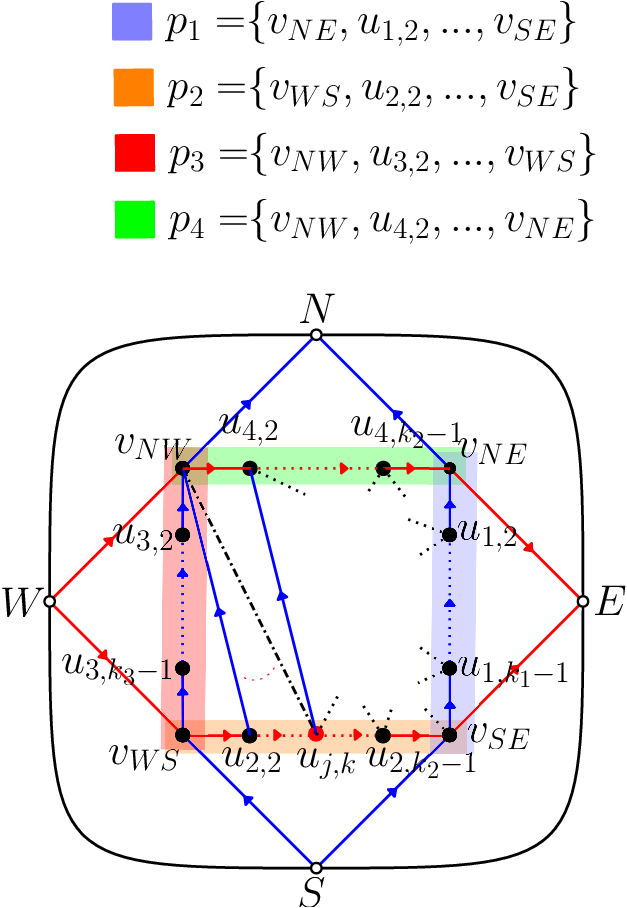}
         \caption{ }
         \label{thpr12}
     \end{subfigure}
     \hspace{0.5cm}
     \begin{subfigure}[b]{0.3\textwidth}
         \centering
         \includegraphics[width=\textwidth]{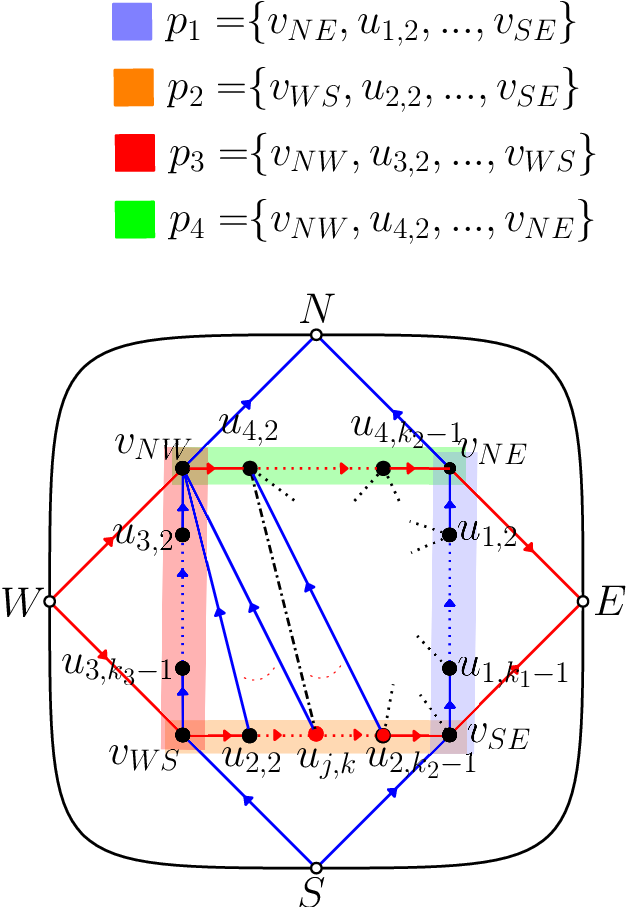}
         \caption{ }
         \label{thpr13}
     \end{subfigure}
        \caption{(a) Existence of a flipable edge either ($v_{x}\ u_{j,k}$), $j=2, x={NW}$ (b) Existence of a flipable edge ($u_{j,k}, u$), $j=2$ and $u=u_{4,2}$.}
        \label{Different4comp4}
\end{figure}
\end{proof}

\section{Methodology I: Construction of Area-Universal Rectangular Layouts}\label{Highercase}

  To address the area–universal layout problem for biconnected outerplanar graphs, we first need a mechanism that identifies whether the given input graph satisfies the necessary structural condition for such a layout to exist. Specifically, a biconnected outerplanar graph admits an area–universal layout only if it contains exactly two vertices of degree two (see Theorem \ref{t1}). This condition can be checked in linear time by simply counting degree-2 vertices.


\noindent
The purpose of \textsc{Outer4Completion} is to transform a given biconnected outerplanar graph into an \emph{extended} outerplanar graph such that the resulting extension admits an area-universal rectangular layout. In particular, the augmentation ensures that every regular edge labeling of the extended graph $E(\mathcal{G})$ contains no flippable edges.

\noindent
\textbf{Notation and helper functions:}
\begin{itemize}
  \item $B=(b_1,\dots,b_n)$: ordered outer-face cycle.
  \item $d_1,d_2$: the two degree-two vertices.
  \item $\text{CARD}=[N,E,S,W]$: array of the four new cardinal vertices.
  \item $\text{indexOf}(v,B)$: returns the position of $v$ in $B$.
  \item $\text{addEdge}(u,w)$: inserts an edge $(u,w)$ while preserving planarity.
  \item For $1 \le i \le 4$, $P_i$ denote the set of vertices selected from the outer cycle $B$, listed in the same ordering as they appear in $B$, and determined by the vertices $v, d_1,$ and $d_2$.

  \item $E_c$: set of edges introduced during the augmentation process to construct $E(\mathcal{G})$.
\end{itemize}

This augmentation procedure is formally described in Algorithm~\ref{algo2}. Given a biconnected outerplanar graph $\mathcal{G}$ with outer--face cycle
$B=(b_1,\ldots,b_n)$, and exactly two degree--2 vertices $d_1$ and $d_2$, the
algorithm augments $\mathcal{G}$ by introducing four new vertices
$N,E,S,W$ arranged in clockwise order along the boundary. These vertices,
collectively denoted as $\text{CARD}=[N,E,S,W]$, serve as the cardinal vertices. A pivot vertex $v\in V$ with degree greater than two
is then selected. The outer boundary cycle $B$ is traversed clockwise starting
from the successor of $v$, and is partitioned into four sets
$P_1,P_2,P_3$,and $P_4$ determined by the locations of $d_1$ and $d_2$ on $B$.
\begin{itemize}
  \item[$P_1$]: consisting solely of the vertex $v$,
  \item[$P_2$]: the vertices encountered from $v$ to the first degree--2 vertex $d_1$
        in clockwise order along $B$,
  \item[$P_3$]: the vertices from $d_1$ to the second degree--2 vertex $d_2$
        in clockwise order along $B$,
  \item[$P_4$]: the remaining vertices from $d_2$ back to $v$.
\end{itemize}
The vertices in $P_i$, $1\leq i \leq 4$ are then made adjacent to their designated cardinal vertices, with vertex $v$ itself adjacent to three consecutive cardinal vertices $(W,N,E)$. Finally, the four cardinal vertices are connected by edges into the four-cycle
$(N,E,S,W)$ so that the new exterior face of the augmented graph is a
quadrilateral. The output is the augmented graph $\mathcal{G}_v$ that serves
as a valid \emph{extended} outerplanar of $\mathcal{G}$, guaranteeing that it admit an area--universal rectangular layout.

\begin{algorithm}[H]
\caption{Outer4Completion($\mathcal G,B,d_1,d_2$)}
\label{algo2}
\begin{algorithmic}[1]\small
\State \textbf{Input:} biconnected outerplanar graph $\mathcal G=(V,E)$, outer-face cycle $B=(b_1,\dots,b_n)$ (clockwise), degree-$2$ vertices $d_1,d_2$
\State \textbf{Output:} augmented graph $\mathcal G_v=(V\cup\{N,E,S,W\},E\cup E_C)$
\State $V_{>2}\gets\{u\in V:\deg_{\mathcal G_p}(u)>2\}$
\State add vertices $N,E,S,W$ on outer boundary; CARD$=[N,E,S,W]$ \Comment{clockwise indices $0..3$}

\State choose a vertex $v\in V_{>2}$
  \State $\mathcal{G}_v\leftarrow \mathcal{G}$
    \State addEdge$(v,W)$, addEdge$(v,N)$, addEdge$(v,E)$
    \State $i\leftarrow$ (indexOf$(v,B)+1)\bmod n$,
  \State $P_2 = \phi$, $P_3 =\phi$, $P_4 =\phi$
  \While{true}
    \State append $B[i]$ to $P_2$; \If{$B[i]=d_1$ or $B[i]=d_2$}
    \State append $B[i]$ to $P_3$
    \State $i\leftarrow (i+1) \bmod n$
    \State \textbf{break} \EndIf
    \State $i\leftarrow(i+1)\bmod n$
  \EndWhile
  \ForAll{$u\in P_2$} addEdge$(u,E)$ \EndFor
  \While{true}
    \State append $B[i]$ to $P_3$; \If{$B[i]=d_1$ or $B[i]=d_2$}
    \State append $B[i]$ to $P_4$
    \State $i\leftarrow (i+1) \bmod n$
    \State \textbf{break} \EndIf
    \State $i\leftarrow(i+1)\bmod n$
  \EndWhile
  \ForAll{$u\in P_3$} addEdge$(u,S)$ \EndFor
  \While{true}
    \State \If{$B[i]=v$} then \textbf{Break} \EndIf
    \State append $B[i]$ to $P_4$
    \State $i\leftarrow (i+1) \bmod n$
  \EndWhile
  \ForAll{$u\in P_4$} addEdge$(u,W)$ \EndFor
    \State addEdge($\mathcal{G}_v,(N,E)$); addEdge($\mathcal{G}_v,(E,S)$); addEdge($\mathcal{G}_v,(S,W)$); addEdge($\mathcal{G}_v,(W,N)$)
    \State \Return $G_v$ \Comment{success: augmented graph}

\end{algorithmic}
\end{algorithm}

\subsection{Correctness and Time Complexity}
\noindent
To establish that the \textsc{Outer4Completion} procedure always produces an
\emph{extended} outerplanar graph that admits an area--universal layout, independent
of the particular chosen REL, we rely on the following
theorem. Specifically, we show that for any biconnected outerplanar graph $\mathcal{G}$,
if there exists a vertex $v$ with $\deg(v) > 2$ that can be made adjacent to three consecutive cardinal vertices, then the resulting \emph{extended} outerplanar graph admit an area--universal layout
corresponding to every possible REL. This claim is formally stated in
Theorem~\ref{sufficient}.

 \begin{theorem}\label{sufficient}
Let $\mathcal{G}$ be a biconnected outerplanar graph of order $n>3$, and let $v \in V(\mathcal{G})$ be a vertex of degree $k > 2$. It is sufficient for $E(\mathcal{G})$ to admit a one-sided rectangular layout (area-universal layout) if $v$ is adjacent to exactly three cardinal vertices in $E(\mathcal{G})$.
 \end{theorem}

 \begin{proof}
     Let $\mathcal{G}$ be a biconnected outerplanar graph embedded in the plane with its vertices listed in clockwise order.  Choose a vertex $v\in V(\mathcal{G})$ of degree more than 2 and make it adjacent to exactly three cardinal vertices; denote the four cardinals in clockwise order by $N,E,S,W$.  Then $V(\mathcal{G})$ is partitioned uniquely into four contiguous boundary paths $P_1,P_2,P_3,P_4$.  The uniqueness follows from outerplanarity together with the fact that $v$ ($deg(v)>2$) is incident to three distinct cardinals (see Figures \ref{th2pr1} and \ref{th2pr2}).  We adopt the convention $P_1=\{v\}$ and take $P_2,P_3,P_4$ in clockwise order from $v$, so that $P_2$ begins at $v$ and ends at the first degree-$2$ vertex $d_1$, $P_3$ runs from $d_1$ to the other degree-$2$ vertex $d_2$, and the remaining boundary vertices (including $d_2$) lie on $P_4$.  Any other partition of the outer face into four paths would produce a separating triangle and hence would not correspond to a valid RFP.

Since $v$ belongs to both $P_2$ and $P_4$, realizing the third (non-cardinal) adjacency of $v$ by connecting it to a vertex of either $P_2$ or $P_4$ would create a forbidden local configuration (a complex/separating triangle).  Therefore, the third neighbor that witnesses $\deg(v)>2$ must lie on $P_3$; denote this vertex by $w\in P_3$.  With this choice, no adjacency between the vertices of $P_2$ and the vertices of $P_4$ is possible (see Figure \ref{th2pr3}). 

\begin{figure}
     \centering
     \begin{subfigure}[b]{0.25\textwidth}
         \centering
         \includegraphics[width=\textwidth]{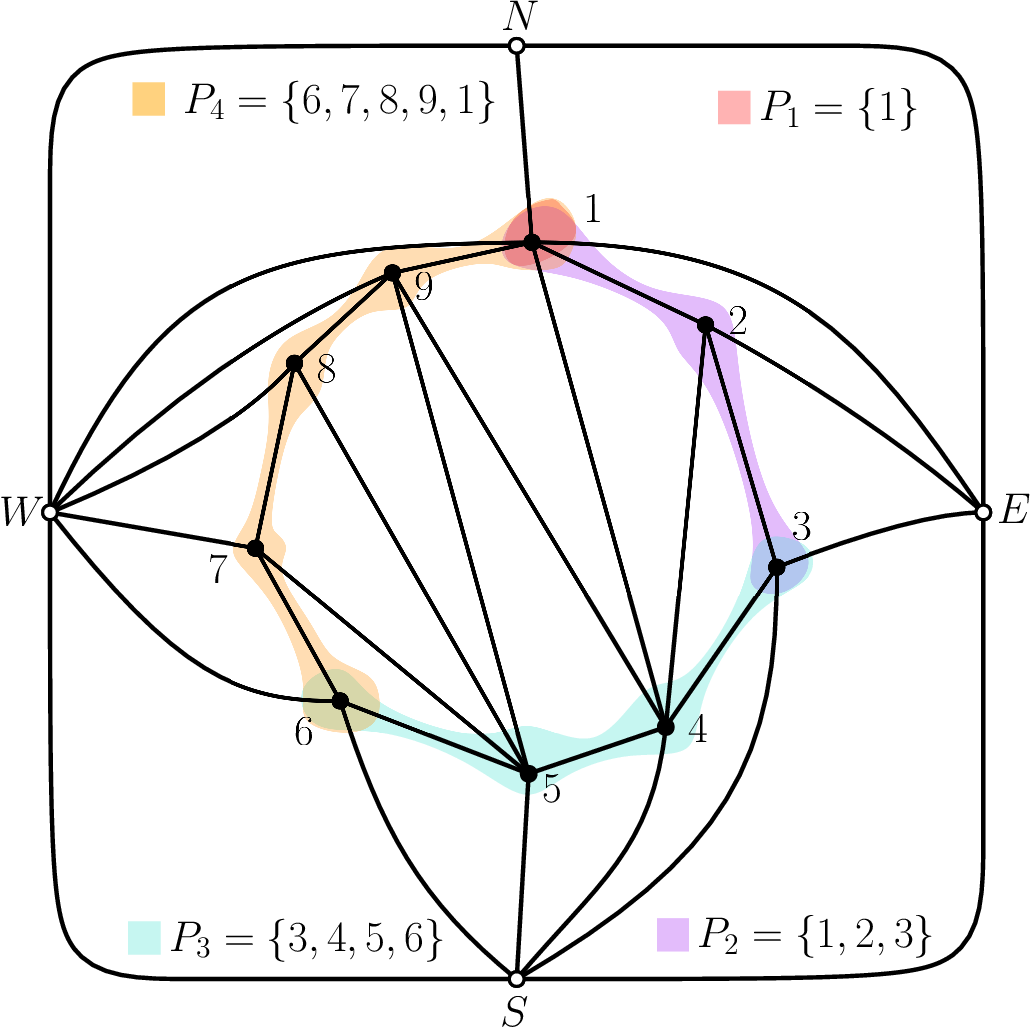}
         \caption{ }
         \label{th2pr1}
     \end{subfigure}
     \hspace{1cm}
     \begin{subfigure}[b]{0.25\textwidth}
         \centering
         \includegraphics[width=\textwidth]{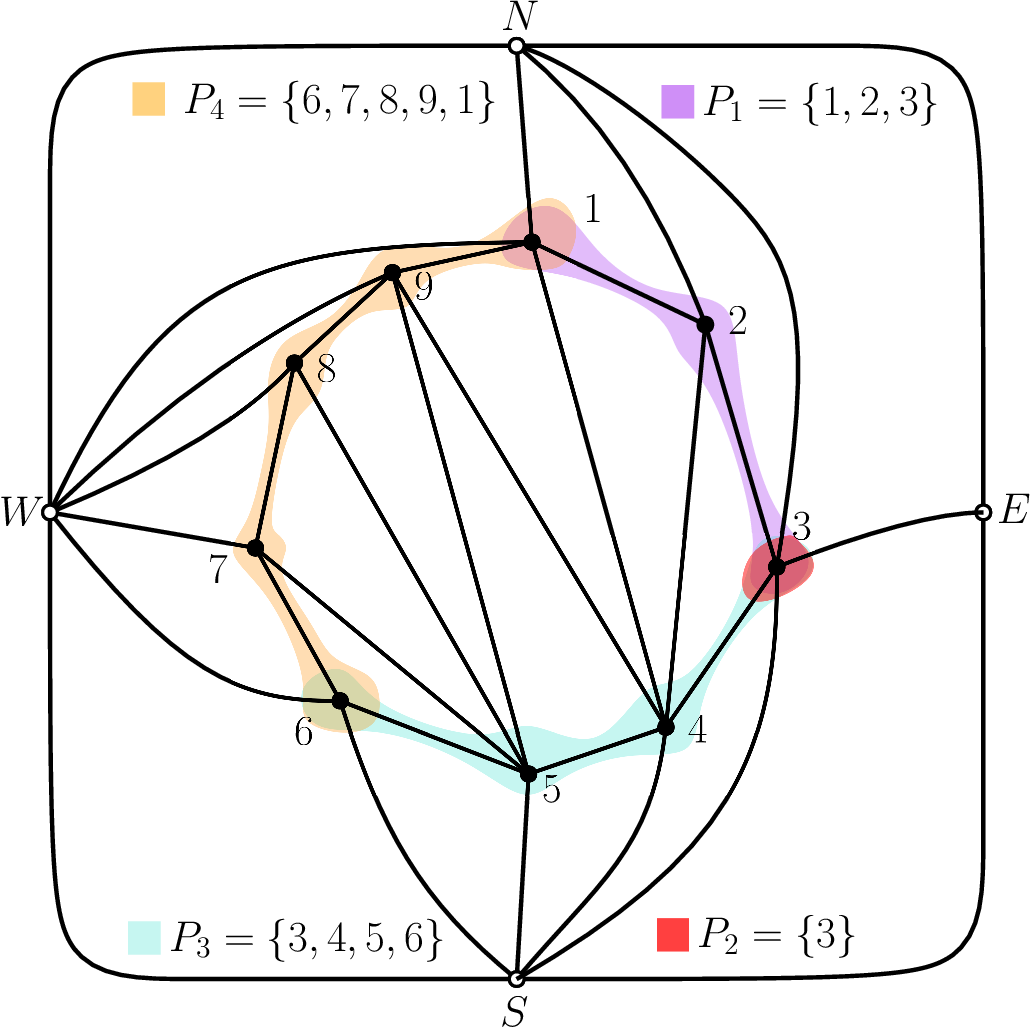}
         \caption{ }
         \label{th2pr2}
     \end{subfigure}
     \hspace{1cm}
     \begin{subfigure}[b]{0.25\textwidth}
         \centering
         \includegraphics[width=\textwidth]{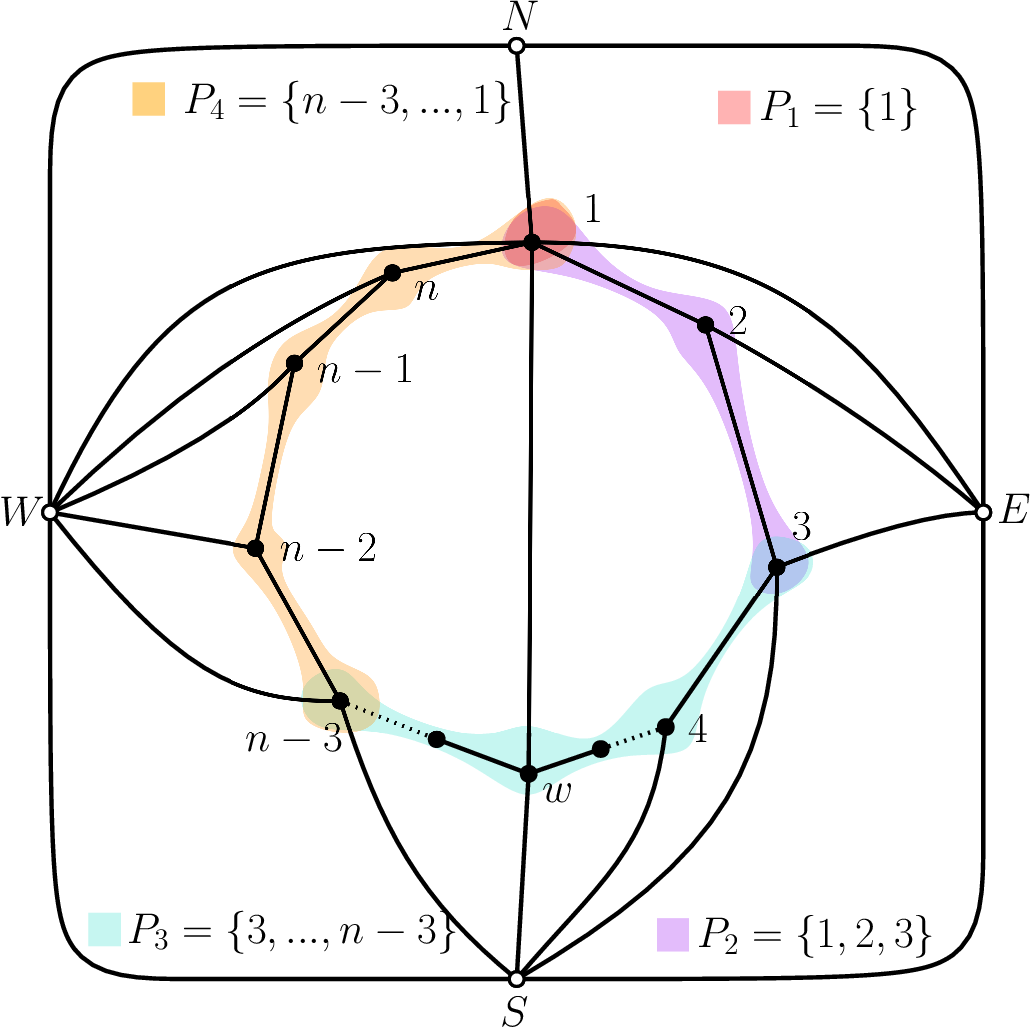}
         \caption{ }
         \label{th2pr3}
     \end{subfigure}

        \caption{(a) Unique partition of paths when vertex $1$ is adjacent to three cardinal vertices. 
        (b) Multiple partitions are possible in an outerplanar graph, even when vertex $3$ is adjacent to three cardinal vertices. 
        (c) The adjacency $(u,v)$ does not exist, where $u \in P_2$ and $v \in P_4$.
        }
        \label{theorem4}
\end{figure}
Now, vertices in the path $P_2$ can be adjacent only to vertices in the path $P_3$, except for one adjacency with vertex in the path $P_1$. Similarly, vertices in the path $P_4$ can be adjacent only to vertices in the path $P_3$, apart from one adjacency with vertex in the path $P_1$. Under this partition, there are only two possible four-cycle configurations: either two vertices from one path $p_i$ and two from another path $p_j$, or one vertex from one path $p_i$ and three from another path $p_j$. In both cases, an alternating four-cycle, that is, a four-cycle whose edges are alternately colored in the regular edge labeling, cannot occur.
Consequently, no flippable edge arises in any REL produced by this construction. Therefore, any REL derived from this augmentation will ultimately yield an area-universal layout.\hfill $\square$
\end{proof}
\noindent
The time complexity of Algorithm~\ref{algo2} can be analyzed as follows. 
Constructing the set $V_{>2}$ of vertices of degree greater than two requires 
a single pass over all vertices of $\mathcal{G}$, which takes $O(n)$ time, 
where $n=|V|$. The algorithm then selects one pivot vertex $v$ and makes a 
single call to \texttt{indexOf$(v,B)$}, which involves scanning the boundary cycle $B$ once and hence costs $O(n)$ time. The subsequent partitioning of the boundary into three contiguous 
sets $P_2,P_3,P_4$ is carried out by advancing a pointer around the outer cycle, and therefore each vertex on $B$ is processed at most once, yielding another $O(n)$ contribution. Attaching vertices of each partition to the 
designated cardinal nodes incurs a constant number of edge insertions per vertex, also bounded by $O(n)$. Finally, the algorithm adds four edges to close the cycle $(N,E,S,W)$, which is constant time. Thus the total running time for 
constructing a four-completion with respect to a single pivot is bounded by 
$O(n)$, dominated by linear scans of the boundary and degree computations.

\section{Methodology II: Construction of Area-Universal Rectangular Layouts}\label{lowercase}

In Section~\ref{Highercase}, we presented an algorithm that constructs a unique \emph{extended} outerplanar graph when a vertex $v \in V(\mathcal{G})$ is adjacent to exactly three cardinal vertices in $E(\mathcal{G})$. However, when a degree-two vertex is adjacent to three cardinal vertices in $E(\mathcal{G})$, multiple distinct extended outerplanar graphs may result from the augmentation, and not all of these admit area-universal rectangular layouts.

In this section, we identify the specific augmentations that must be selected, in the case where a degree-two vertex is adjacent to three cardinal vertices, to obtain an \emph{extended} outerplanar graph that admits an area-universal rectangular layout regardless of the choice of regular edge labeling. As a consequence, we enumerate the total number of area-universal rectangular layouts that realize a given biconnected outerplanar graph.

\begin{theorem}
    Let $\mathcal{G}$ be a biconnected outerplanar graph with exactly two vertices of degree two, and let $E(\mathcal{G})$ denote its extended outerplanar graph. Let $d$ be a degree-two vertex of $\mathcal{G}$ that is adjacent, in $E(\mathcal{G})$, to exactly three cardinal vertices. Then $E(\mathcal{G})$ admits an area-universal rectangular layout if and only if at least one neighbor of $d$ is adjacent to exactly two cardinal vertices in $E(\mathcal{G})$.
\end{theorem}

\begin{proof}
    Let $\mathcal{G}$ be a biconnected outerplanar graph of order $n>4$ with exactly two vertices of degree two, and let $E(\mathcal{G})$ denote its extended outerplanar graph obtained by adding the four cardinal vertices $N, W, S,$ and $E$ on the outer face in clockwise order. Suppose that a degree-two vertex $d \in V(\mathcal{G})$ is adjacent, in $E(\mathcal{G})$, to exactly three cardinal vertices. Without loss of generality, assume that $d$ is adjacent to the cardinal vertices $E$, $N$, and $W$. Let $w_1$ and $w_2$ denote the two neighbors of $d$ in the graph $\mathcal{G}$.

    We first establish a necessary condition: if both neighbors of $d$ are adjacent to exactly one cardinal vertex, then multiple distinct \emph{extended} outerplanar graphs can be obtained from the augmentation (see Figure \ref{th2pr2}), and none of these admit an area-universal layout, regardless of the choice of regular edge labeling. We then prove that this condition is also sufficient. Specifically, if at least one neighbor of $d$ is adjacent to exactly two cardinal vertices, then the augmentation yields a unique extended outerplanar graph $E(\mathcal{G})$ that admits an area-universal rectangular layout.

    Since the augmentation used to construct $E(\mathcal{G})$ partitions the vertex set
$V(\mathcal{G})$ into four boundary paths $(P_0, \allowbreak P_1, P_2, P_3)$, assume without loss
of generality that the path $P_0$ consists only of the degree-two vertex $d$ (as $d$ is adjacent to three cardinal vertices). Suppose that both neighbors $w_1$ and $w_2$ of $d$ are adjacent to exactly one cardinal vertex (see Figure~\ref{degree2thm1}). In this case, each of $w_1$ and $w_2$ lies on exactly
one boundary path, with $w_1 \in P_3$ and $w_2 \in P_1$. Under this configuration, multiple distinct \emph{extended} outerplanar graphs can be
obtained from the augmentation, with the number of possible extensions depending on
$|V(\mathcal{G})|$. Let $v_1, v_2 \in V(\mathcal{G})$ be two adjacent vertices such that
either $v_1 \in N(w_1) \cap N(w_2)$ or $v_2 \in N(w_1) \cap N(w_2)$ (see Figures \ref{degree2thm2} and \ref{degree2thm2}). Since
$(w_1, w_2) \in T_2$, $(v_1, w_1) \in T_1$, $(v_2, w_2) \in T_1$, and
$(v_1, v_2) \in T_2$ (as $v_1 \in P_3$ and $v_2 \in P_1$), these edges together form an
alternating four-cycle $(w_1, w_2, v_2, v_1)$. If $v_1 \in N(w_1) \cap N(w_2)$, then the edge $(v_1, w_2)$ is flippable. Similarly, if
$v_2 \in N(w_1) \cap N(w_2)$, then the edge $(w_1, v_2)$ is flippable. Consequently,
every such augmentation introduces at least one flippable edge in the corresponding
regular edge labeling.

\begin{figure}
     \centering
     \begin{subfigure}[b]{0.25\textwidth}
         \centering
         \includegraphics[width=\textwidth]{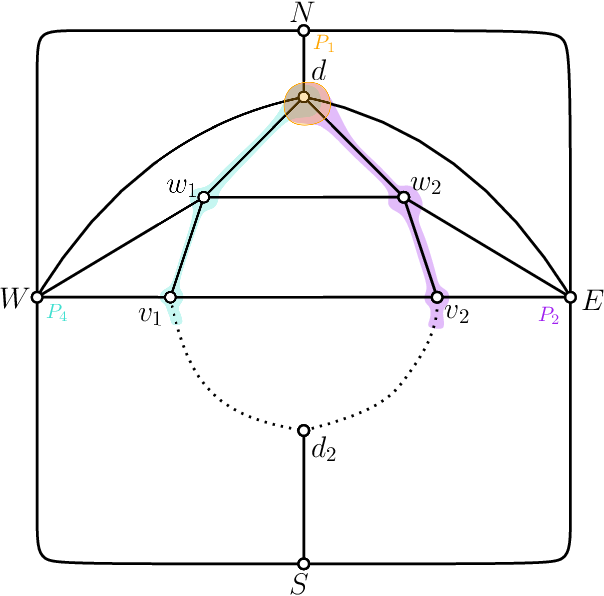}
         \caption{ }
         \label{degree2thm1}
     \end{subfigure}
     \hspace{1cm}
     \begin{subfigure}[b]{0.25\textwidth}
         \centering
         \includegraphics[width=\textwidth]{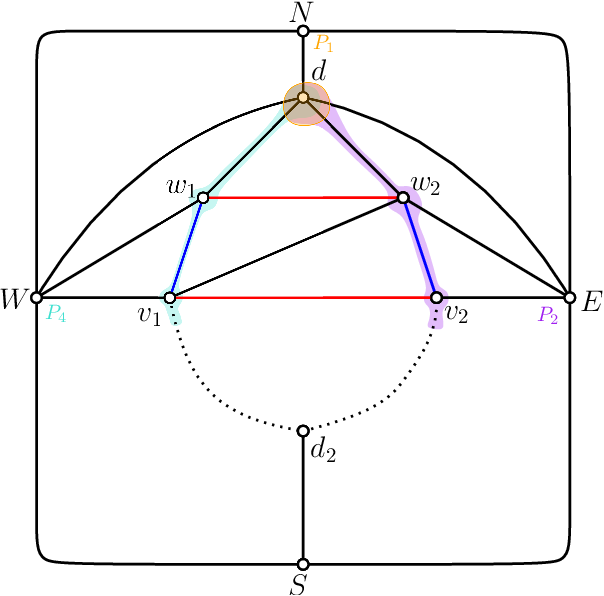}
         \caption{ }
         \label{degree2thm2}
     \end{subfigure}
     \hspace{1cm}
     \begin{subfigure}[b]{0.25\textwidth}
         \centering
         \includegraphics[width=\textwidth]{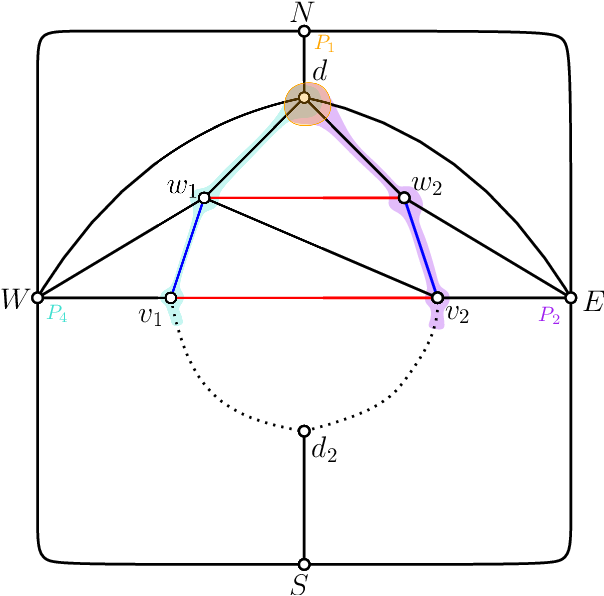}
         \caption{ }
         \label{degree2thm3}
     \end{subfigure}
     \caption{(a) Multiple \emph{extended} outerplanar graph possible when $|V(\mathcal{G})|>4$. (b,c) A flippable edge exist in the alternate four-cyle $(w_1-w_2-v_2-v_1).$}
\end{figure}

On the other hand, suppose that one of the neighbors of $d$ is adjacent to exactly two
cardinal vertices; without loss of generality, let $w_2$ be adjacent to the cardinal
vertices $N$ and $E$. In this case, the augmentation yields a unique \emph{extended}
outerplanar graph. The induced partition of the boundary vertices consists of the four
paths $P_0 = \{d\}, P_1 = \{d, w_2\}, P_2 = \{w_2, \ldots, d_2\},$ and $P_3 = \{d_2, \ldots, d\}.$ Let $C = (a_0,a_1,a_2,a_3)$ be any four-cycle in the augmented graph $E(\mathcal{G})$.
For $C$ to be an alternating four-cycle, there must exist indices
$k \in \{0,1,2,3\}$ and $i \in \{0,1,2,3\}$ (indices taken modulo $4$) such that
\[
\{a_i,a_{i+1}\} \subseteq P_k
\quad \text{and} \quad
\{a_{i+2},a_{i+3}\} \subseteq P_{k+2},
\]
with
\[
\{a_i,a_{i+1}\} \cap P_{k+2} = \emptyset
\quad \text{and} \quad
\{a_{i+2},a_{i+3}\} \cap P_k = \emptyset .
\]

Under the given partition of the boundary vertices
$(P_0,P_1,P_2,P_3)$, no four-cycle $C$ satisfies these conditions. Hence, $E(\mathcal{G})$ contains no alternating four-cycle. It follows that there does not exist a flippable edge. Therefore, the rectangular floor plan corresponding to $E(\mathcal{G})$ is area-universal. \hfill $\square$
\end{proof}

\begin{algorithm}[H]
\caption{degree-2 extension($\mathcal{G},B,d_1,d_2$)}
\label{algodegree-2}
\begin{algorithmic}[1]\small
\State \textbf{Input:} biconnected outerplanar graph $\mathcal{G}=(V,E)$,
outer-face cycle $B=(b_1,\dots,b_n)$ (clockwise),
degree-$2$ vertices $d_1,d_2$
\State \textbf{Output:} set $\mathcal{S}$ of augmented graphs

\State $\mathcal{S}\gets\emptyset$
\State add cardinal vertices $N,E,S,W$ in clockwise order

\ForAll{$d \in \{d_1,d_2\}$} \Comment{choose degree-2 vertex}
  \State $i_d \gets \text{indexOf}(d,B)$
  \State $w^{-} \gets B[(i_d-1)\bmod n]$ \Comment{neighbor before $d$}
  \State $w^{+} \gets B[(i_d+1)\bmod n]$ \Comment{neighbor after $d$}

  \ForAll{$w \in \{w^{-},w^{+}\}$} \Comment{choose neighbor with two cardinal adjacencies}
    \State $\mathcal{G}_{d,w}\leftarrow\mathcal{G}$
    \State addEdge$(d,W)$; addEdge$(d,N)$; addEdge$(d,E)$
    \If{$w = w^{+}$} \Comment{$w$ immediately after $d$}
      \State addEdge$(w,E)$; addEdge$(w,S)$
      \State $P_0\gets\{d\}$,\quad $P_1\gets\{d,w\}$, $P_2\gets \{w\}$, $P_3\gets\emptyset$
      \State $i\gets(\text{indexOf}(w,B)+1)\bmod n$
      \While{$B[i]\neq w^{-}$}
        \State append $B[i]$ to $P_2$
        \State $i\gets(i+1)\bmod n$
      \EndWhile
      \State append $B[i]$ to $P_2$ and $P_3$
      \State $i\gets(i+1)\bmod n$
      \While{$B[i]\neq d$}
        \State append $B[i]$ to $P_3$
        \State $i\gets(i+1)\bmod n$
      \EndWhile
      \State append $B[i]$ to $P_3$
  \Else \Comment{$w$ immediately before $d$}
      \State addEdge$(w,W)$; addEdge$(w,S)$
      \State $P_0\gets\{d\}$,\quad $P_3\gets\{d,w\}$, $P_2\gets \{w\}$, $P_1\gets\emptyset$
      \State $i\gets(\text{indexOf}(w,B)-1)\bmod n$
      \While{$B[i]\neq w^{-}$}
        \State append $B[i]$ to $P_2$
        \State $i\gets(i-1)\bmod n$
      \EndWhile
      \State append $B[i]$ to $P_2$ and $P_1$
      \State $i\gets(i-1)\bmod n$
      \While{$B[i]\neq d$}
        \State append $B[i]$ to $P_3$
        \State $i\gets(i-1)\bmod n$
      \EndWhile
      \State append $B[i]$ to $P_1$
    \EndIf

    \ForAll{$u\in P_1$} addEdge$(u,E)$ \EndFor
    \ForAll{$u\in P_2$} addEdge$(u,S)$ \EndFor
    \ForAll{$u\in P_3$} addEdge$(u,W)$ \EndFor

    \Comment{connect cardinal vertices}
    \State addEdge$(N,E)$; addEdge$(E,S)$
    \State addEdge$(S,W)$; addEdge$(W,N)$

    \State $\mathcal{S}\gets\mathcal{S}\cup\{\mathcal{G}_{d,w}\}$
  \EndFor
\EndFor

\State \Return $\mathcal{S}$
\end{algorithmic}
\end{algorithm}

Kant and He~\cite{kant1997regular} proved that for any \emph{extended} graph $E(\mathcal{G})$,
a regular edge labeling can be computed in linear time, and that the corresponding
rectangular floor plan defined by this labeling can also be constructed in linear time.
However, for the particular \emph{extended} outerplanar graph produced by
Algorithm~\ref{algodegree-2}, it is not necessary to compute a regular edge labeling.
Instead, we can directly construct the rectangular floorplan by placing the vertices
on an $n \times n$ grid using the Algorithm described below. The resulting layout is
always area-universal.

\begin{algorithm}[H]
\caption{AURFP($\mathcal G,d,w,P_1$)}
\label{algoAURFP}
\begin{algorithmic}[1]\small
\State \textbf{Input:} A biconnected outerplanar graph with designated degree-2 vertex and its neighbour and a valid partition of $V(\mathcal{G})$
\State \textbf{Output:} An area-universal layout
\State $L \gets [d, w]$
\State $\ell \gets d$, $\ell^{'} \gets w$ 
\State Let $u \in N_{\mathcal{G}}(\ell \cap \ell^{'})$, append $u$ to $L$
\State $\text{Visited} \gets \{d, w, u\}$

\While{$|\text{Visited}| < |V(\mathcal{G}|$}
    \State $A = N_{\mathcal{G}}(u,\ell) \setminus Visited$, $B = N_{\mathcal{G}}(u,\ell^{'}) \setminus Visited$
    \If{$A \neq \phi$}
    \State choose $v\in A$ and append $v$ to $L$
    \State Visited $\gets$ Visited $\cup \{v\}$
    \State $\ell^{'} \gets u$ and $u \gets v$
    \ElsIf{$B \neq \phi$}
    \State choose $v\in B$ and append $v$ to $L$
    \State Visited $\gets$ Visited $\cup \{v\}$
    \State $\ell \gets u$ and $u \gets v$
    \EndIf
\EndWhile
\State Draw a block and name it $L[0]$
\State Construct an $n \times n$ grid within the block
\For{$i=2$ to $|V(\mathcal{G})|$}
\If{$L[i-2]\in P_1$}
\State Slice the bottom-right block by a horizontal line segment at $i^{th}$ level; Name the new block $L[i-1]$
\Else
\State Slice the bottom-right block by a vertical line segment at $i^{th}$ level; Name the new block $L[i-1]$
\EndIf
\EndFor

\end{algorithmic}
\end{algorithm}

\begin{figure}
     \centering
     \begin{subfigure}[b]{0.35\textwidth}
         \centering
         \includegraphics[width=\textwidth]{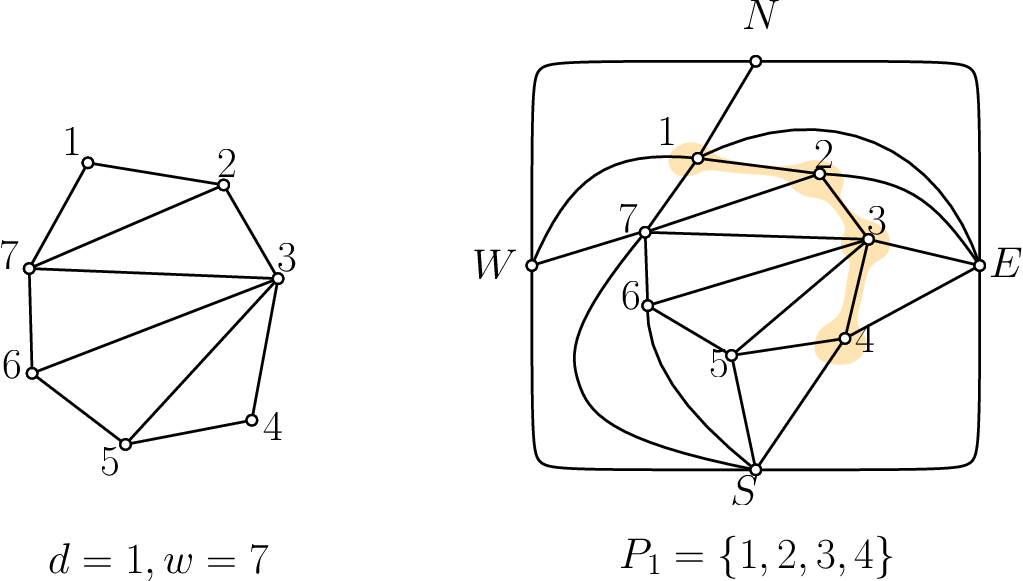}
         \caption{Augmentation of an outerplanar graph to obtain an extended outerplanar graph using Algorithm~\ref{algodegree-2} }
         \label{degree2algo1}
     \end{subfigure}
     \hspace{1cm}
     \begin{subfigure}[b]{0.5\textwidth}
         \centering
         \includegraphics[width=\textwidth]{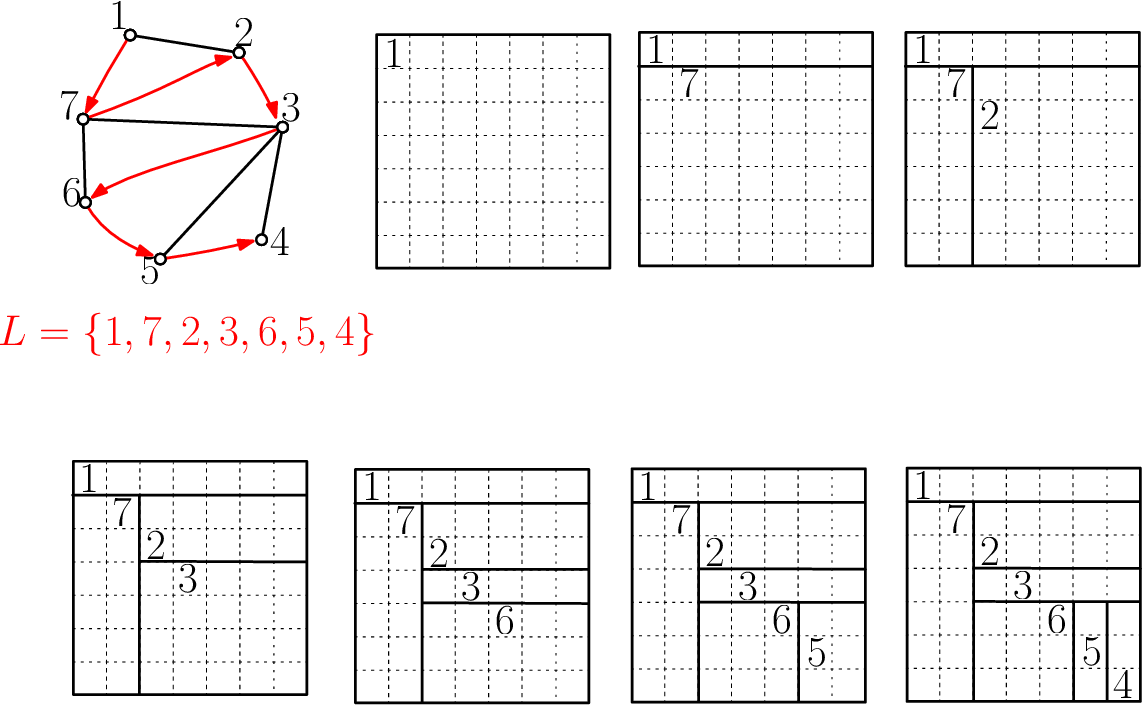}
         \caption{Generating an area-universal layout using Algorithm \ref{algoAURFP}. }
         \label{degree2algo2}
     \end{subfigure}
     \caption{Illustration of Algorithms \ref{algodegree-2} and \ref{algoAURFP}, showing the construction of an area-universal rectangular layout by fixing a degree-two vertex adjacent to three cardinal vertices.}
     \label{fig:degree2algo}
\end{figure}
Algorithm~\ref{algoAURFP} always generates an area-universal rectangular layout. At each
iteration $i$, a new block is created by inserting a maximal line segment that forms a
side of the rectangle corresponding to $L[i-2]$. Consequently, every maximal line
segment introduced by the algorithm is one-sided. Hence, the resulting rectangular layout is area-universal.

\section{Enumeration of Area-Universal Rectangular Layouts Realizing the Adjacency Structure of a Biconnected Outerplanar Graph}
In this paper, we characterize the class of \emph{extended} outerplanar graphs obtained by augmenting a biconnected outerplanar graph with exactly two vertices of degree two that admit area-universal rectangular layouts. In Section~\ref{necessarysection}, we show that if no vertex in the extended outerplanar graph is adjacent to exactly three cardinal vertices, then the corresponding rectangular layout is not area-universal. We then prove, in section \ref{Highercase}, that when a vertex is adjacent to exactly three cardinal vertices, this condition is sufficient to obtain an area-universal rectangular layout only if the chosen vertex has degree greater than two in the original graph; in this case, the augmentation yields a unique \emph{extended} outerplanar graph. In contrast, when the chosen vertex has degree two, the augmentation may result in multiple distinct \emph{extended} outerplanar graphs, not all of which admit area-universal rectangular layouts.

Thus, when a vertex of degree at least three is chosen and made adjacent to exactly three
cardinal vertices, the resulting \emph{extended} outerplanar graph is unique and admits a
unique regular edge labeling. Consequently, for each such chosen vertex, there exists
precisely one area-universal rectangular arrangement, up to rotation.

In Section~\ref{lowercase}, we show that when a degree-two vertex is adjacent to exactly
three cardinal vertices, the augmentation required to obtain an extended graph that
admits an area-universal rectangular layout depends on the relative positions of its
neighbors. In particular, if one of the neighbors of the degree-two vertex is adjacent
to exactly two cardinal vertices, then the resulting \emph{extended} outerplanar graph
is uniquely determined and admits an area-universal layout. For this case, exactly four
distinct extended outerplanar graphs admit area-universal rectangular layouts.

In conclusion, there are exactly $|V(\mathcal{G})|+2$ area-universal rectangular layouts
(up to rotation) that realize a given biconnected outerplanar graph.


\section{Conclusion and Future Work}

\noindent
In this work we have addressed the problem of constructing area--universal
layouts for biconnected outerplanar graphs. Our main contribution is the identification of
both necessary and sufficient conditions under which a biconnceded outerplanar graph can be augmented to obtain \emph{extended} outerplanar graph that admits area-universal layout. In particular, we established that while the presence of
exactly two degree--two vertices in \emph{proper} graph is a necessary requirement. However, not every augmentation of such a graph leads to an area-universal layout. To overcome this limitation, we characterized \emph{extended} outerplanar graphs based on the adjacencies of vertices of \emph{proper} graph and \emph{cardinal} vertices.  (see Figures \ref{fig:degree2algo} and \ref{front}).
\begin{figure}
    \centering
    \includegraphics[width=0.5\linewidth]{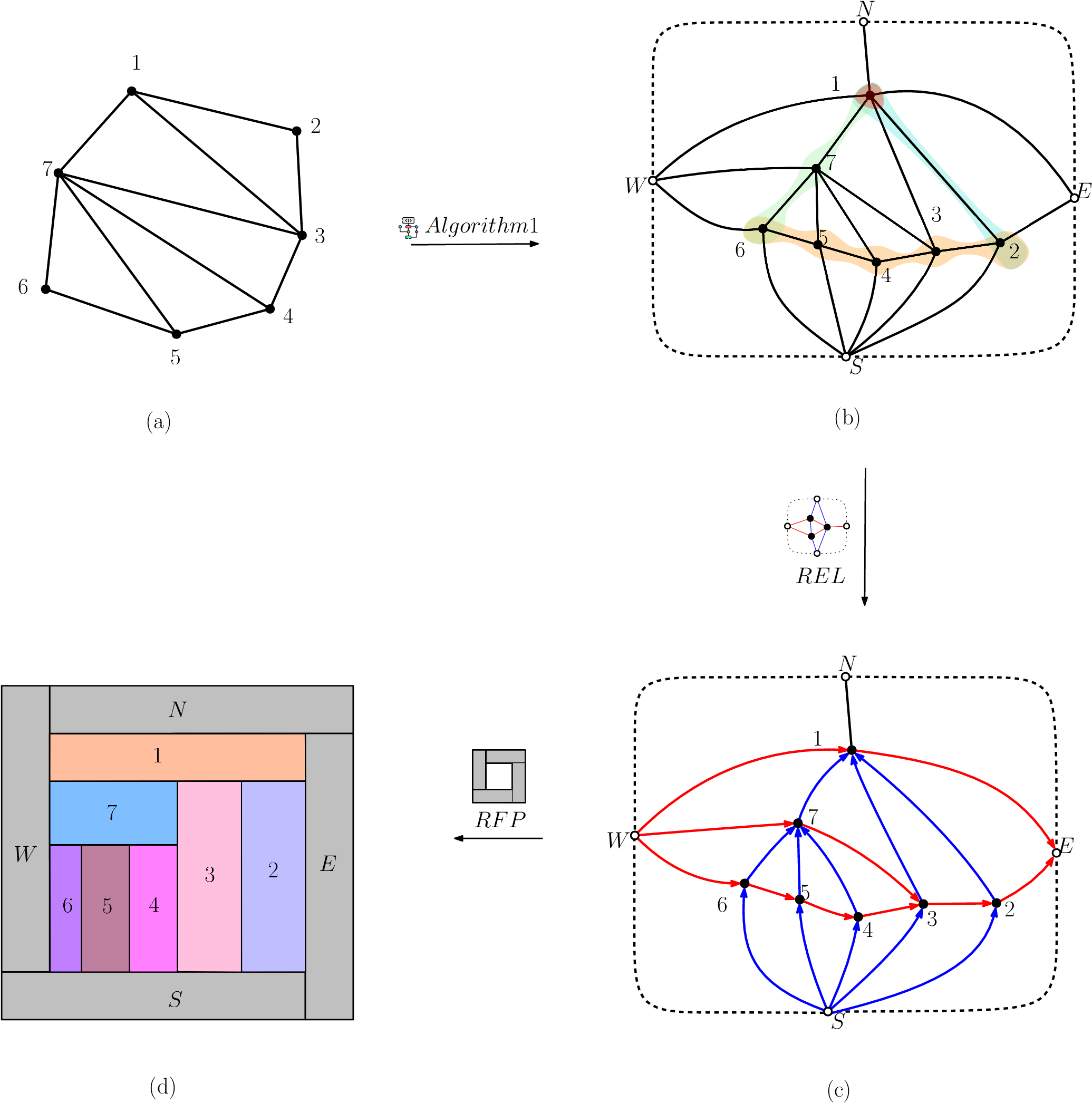}
    \caption{(a) An outerplanar graph with exactly two vertices of degree-2, (b) An extended outerplanar graph, with four contiguous paths separated with different colors, obtained from Algorithm \ref{algo2}, (c)  A REL of the extended outerplanar graph, (d) An area-universal layout corresponding to input outerplanar graph.}
    \label{front}
\end{figure}

The results presented here provide a structural foundation for constructing area--universal layouts in
outerplanar graphs. There are several potential avenues for future research based on this work. A significant direction would be to generalize the concept of area-universal layouts to a broader class of graphs. Specifically, identifying and characterizing the entire class of graphs for which an area-universal layout exists remains an open problem. Additionally, developing efficient algorithms to construct such layouts for these graphs is another challenging yet promising area for future exploration. Extending the approach to handle more complex graph structures and refining the algorithm to optimize layout properties such as compactness and symmetry can further enhance its applicability to real-world problems.

\bibliographystyle{splncs04}
\bibliography{mybib} 

\end{document}